\documentclass[11pt,letterpaper]{article}
\usepackage[margin=1in]{geometry}

\usepackage{etoolbox}
\usepackage{style}
\usepackage{shortcuts}
\usetikzlibrary{matrix, positioning}

\title{Deterministic 3SUM-Hardness}
\author{Nick Fischer\thanks{Weizmann Institute of Science. This work is part of the project CONJEXITY that has received funding from the European Research Council (ERC) under the European Union's Horizon Europe research and innovation programme (grant agreement No.~101078482).} \and Piotr Kaliciak\thanks{Jagiellonian University in Kraków.} \and Adam Polak\thanks{Max Planck Institute for Informatics and Jagiellonian University in Kraków.}}
\date{}

\begin{document}

\maketitle
\begin{abstract}
\noindent
As one of the three main pillars of fine-grained complexity theory, the 3SUM problem explains the hardness of many diverse polynomial-time problems via fine-grained reductions. Many of these reductions are either directly based on or heavily inspired by Pătraşcu's framework involving additive hashing and are thus \emph{randomized}. Some selected reductions were derandomized in previous work~[Chan, He; SOSA'20], but the current techniques are limited and a major fraction of the reductions remains randomized.

In this work we gather a toolkit aimed to derandomize reductions based on additive hashing. Using this toolkit, we manage to derandomize \emph{almost all} known 3SUM-hardness reductions. As technical highlights we derandomize the hardness reductions to (offline) Set Disjointness, (offline) Set Intersection and Triangle Listing---these questions were explicitly left open in previous work~[Kopelowitz, Pettie, Porat; SODA'16]. The few exceptions to our work fall into a special category of recent reductions based on structure-versus-randomness dichotomies.

We expect that our toolkit can be readily applied to derandomize future reductions as well. As a conceptual innovation, our work thereby promotes the theory of \emph{deterministic 3SUM-hardness}.

As our second contribution, we prove that there is a \emph{deterministic universe reduction} for 3SUM. Specifically, using additive hashing it is a standard trick to assume that the numbers in 3SUM have size at most $n^3$. We prove that this assumption is similarly valid for \emph{deterministic} algorithms. 
\end{abstract}

\setcounter{page}{0}
\thispagestyle{empty}
\newpage

\section{Introduction}
In this paper we revisit the famous integer 3SUM problem, which involves checking for a given set $A$ of $n$ integers whether there is a triple $a, b, c \in A$ with~\makebox{$a + b = c$}. It is well-known that the 3SUM problem can be solved in deterministic time $\Order(n^2)$ by a classic textbook algorithm, and it is an infamous open problem whether the running time can be substantially improved beyond quadratic time. While recent improvements only scored log-shavings~\cite{BaranDP08,GronlundP18,Freund17,GoldS17,KaneLM19,Chan20}, it became a popular conjecture that there is no truly subquadratic algorithm (i.e., an algorithm running in time $\Order(n^{2-\epsilon})$ for some $\epsilon > 0$).

Over the course of the last decades, this conjecture has evolved into one the three main pillars of fine-grained complexity, explaining the hardness of a diverse set of polynomial-time problems via fine-grained reductions from the 3SUM problem. This trend was initiated in the 1990's by Gajentaan and Overmars~\cite{GajentaanO95} leading to a series of quadratic-time lower bounds for various geometric problems~\cite{GajentaanO95,BergGO97,ArkinCHMSSY98,BoseKT98,Erickson99,AbellanasHIKLMPS01,BarequetH01,SossEO03,ArchambaultEK05,EricksonHM06,CheongEH07,AronovH08}. Later, in~2010, Pătraşcu achieved another breakthrough~\cite{Patrascu10}. Based on earlier ideas due to Baran, Demaine and Pătraşcu~\cite{BaranDP08}, Pătraşcu realized that the interplay between 3SUM and \emph{additive hashing} can be exploited to derive \emph{randomized reductions} from 3SUM to the related \emph{Convolution 3SUM} and \emph{Triangle Listing} problems. This result inspired an even bigger wave of 3SUM-based lower bounds for various combinatorial problems. Today it is known that the class of 3SUM-hard problems encompasses graph problems~\cite{AbboudL13,WilliamsW13,AbboudBKZ22,AbboudBF23,JinX23}, string problems~\cite{AbboudLW14,AmirCLL14}, dynamic problems~\cite{Patrascu10,KopelowitzPP16,Dahlgaard16,AbboudWY18,ChenDGWXY18} and many more.

\subparagraph{Randomized versus Deterministic.}
As noted before, many of these reductions (for non-geometric problems) use Pătraşcu's seminal reduction as a stepping stone, or at least reuse it's ideas. Consequently, a major fraction of the 3SUM-based hardness results are only achieved via \emph{randomized} reductions. These reductions thus imply lower bounds only if we are willing to assume that the 3SUM problem requires quadratic time also for randomized algorithms.

Chan and He~\cite{ChanH20} took a first step in understanding the role of randomness in 3SUM-based hardness reductions. They proposed to replace the hash function in Pătraşcu's reduction from 3SUM to Convolution 3SUM with a deterministic (and simpler) one. This implied that Convolution 3SUM is hard even for deterministic reductions, and entails that there is no deterministic subquadratic algorithm for Convolution 3SUM under the following hypothesis:

\begin{hypothesis}[Deterministic 3SUM] \label{hyp:det-3sum}
For any $\epsilon > 0$ there is a constant $c$, such that there is no deterministic algorithm for the 3SUM problem over $[n^c]$ in time $\Order(n^{2-\epsilon})$.
\end{hypothesis}

While their result entails derandomizations for some reductions~\cite{BringmannN21b,ChanH19,AmirCLL14,WilliamsW13}, for technical reasons it fails to address the full range of problems. Let us describe the issue in a few words. In essence, Chan and He propose to select a hash function $h(x) = x \bmod m$, where $m$ is the product of small primes, and to select these primes one by one using the method of conditional expectations. With this idea one can easily construct a hash function~$h$ with few \emph{collisions} (\makebox{$h(a) = h(b)$}) since the number of collisions is easily computable. However, for the majority of reductions we need the stronger property that there are few \emph{false positives} ($h(a) + h(b) = h(c)$) which is more difficult to test. In the overview in \cref{sec:toolkit} we elaborate on this in more detail.

\subparagraph{Theory of Deterministic 3SUM-Hardness.}
Nevertheless, Chan and He's result marks an important first step towards a \emph{theory of deterministic 3SUM-hardness.} Aiming at a more complete theory, in this work we study the following wide-open question:
\begin{center}
\medskip
\emph{Question 1: Can we derandomize all reductions from 3SUM?}
\medskip
\end{center}
Understanding the power of randomness is one of the central goals in theoretical computer science. However, in fine-grained complexity (and also parameterized complexity) it is common to neglect this aspect and allow for randomized algorithms and reductions by default. In light of Question~1 it would be exciting to see whether at least for 3SUM-hard problems we could gain some foundational knowledge about the necessity of randomness. 

Answering Question 1 would also be particularly important if there turned out to be faster randomized than deterministic algorithms for 3SUM. This is very well possible---in fact, in the history of state-of-the-art 3SUM algorithms (with mildly subquadratic running times) the gap between randomized and deterministic algorithms was only recently closed~\cite{Chan20}.

Of course the net of 3SUM-based reductions is vast, so where should we begin to tackle Question~1? A reasonable starting point is the aforementioned Triangle Listing problem, which is one of the bottlenecks for the many still-randomized reductions. Recall that Pătraşcu~\cite{Patrascu10} proved that Triangle Listing is hard under randomized reductions, and this reduction was later refined by Kopelowitz, Pettie and Porat~\cite{KopelowitzPP16}. They proposed, as appropriate intermediate problems, the (offline) Set Disjointness and (offline) Set Intersection problems, and proved that these problems are 3SUM-hard under randomized reductions. In order to answer Question 1, we should therefore focus on Set Disjointness and Set Intersection. In their paper~\cite{KopelowitzPP16}, Kopelowitz et al.\ remark that ``it would be surprising if [their] construction could be efficiently derandomized''. 

\bigskip
\noindent
We also consider another question related to the role of randomness for the 3SUM problem: universe reductions. It is well known that for randomized algorithms we can assume without loss of generality that the input is over the universe $[\Order(n^3)]$ (simply replace each input $a \in A$ by $h(a)$ where~\makebox{$h : \Int \to [\Order(n^3)]$} is an appropriate additive hash function). For deterministic algorithms, on the other hand, there are no similar bounds. Chan and He~\cite{ChanH20} end their paper with the following open problem:
\begin{center}
\medskip
\emph{Question 2: Can 3SUM for arbitrary integers be reduced deterministically}\\
\emph{to 3SUM for integers bounded by~$n^{O(1)}$?}
\medskip
\end{center}
Besides being an interesting question in its own right, obtaining a deterministic universe reduction would also constitute a handy tool in the design of deterministic 3SUM-based reductions, and thus contribute to our dream goal of a complete theory of deterministic 3SUM-hardness.





\subsection{Our Results}
Our main technical contribution is that we gather a toolset aimed to derandomize 3SUM-based reductions. It consists of two major tools: A deterministic algorithm to compute an additive hash function with few collisions (see \cref{lem:hashing-few-pseudo-solutions}), and a refined deterministic self-reduction for 3SUM (see \cref{lem:dominance}). We establish these tools in \cref{sec:toolkit} and give a short technical overview. Notably, none of these results use heavy black-box derandomization machinery, but instead are rather simple algorithms tailored to the 3SUM problem.

To our great surprise, by a combination of these tools we can completely answer Question~2, and answer Question~1 for \emph{almost all} known 3SUM-based reductions.

\subparagraph{Deterministic Universe Reduction.}
Let us start with our second driving question. We prove that the universe size can indeed be reduced to $n^{\Order(1)}$:

\begin{restatable}[Deterministic Universe Reduction for 3SUM]{theorem}{thmunivreduction} \label{thm:univ-reduction}
If 3SUM over~$[n^3]$ can be solved in deterministic time $\Order(n^{2-\epsilon})$ for some $\epsilon > 0$, then 3SUM over $[U]$ can be solved in deterministic time $\Order(n^{2-\epsilon'} \log^c U)$ for some constants $\epsilon', c > 0$.
\end{restatable}

It is satisfactory that the precise polynomial bound is $n^3$, which matches what is known in terms of randomized algorithms precisely. Interestingly though, we obtain our universe reduction not as a preprocessing step, but rather as a fine-grained reduction from 3SUM with large universe to 3SUM with small universe.

Let us note that Chan and He's original motivation to this question was replacing the dependence on $U$ by a dependence on $n$ (in a reduction from 3SUM to Convolution 3SUM)~\cite{ChanH20}, which our result fails to achieve, because it involves $\log U$ factors in the running time. However, it is a typical assumption in fine-grained complexity theory that the problems involve weights in~$[-n^c, n^c]$ (for some constant $c$), and in this setting the overhead is only polylogarithmic in $n$. In particular, \cref{thm:univ-reduction} implies that unless the deterministic 3SUM hypothesis (\cref{hyp:det-3sum}) fails, there is no deterministic subquadratic-time algorithm for 3SUM over $[n^3]$. In all deterministic 3SUM-hardness reductions we can thus assume without loss of generality that the given 3SUM instance is over $[n^3]$. And indeed, for some problems~\cite{AbboudBBK20} this assumption is necessary to recover the same lower bounds that are known via randomized reductions.

\subparagraph{Derandomizing Almost All 3SUM-Based Reductions.}
Let us turn to our driving Question~1. Up to very few exceptions, we indeed manage to derandomize all remaining randomized 3SUM-based reductions. Our flagship result here is that we derandomize the reductions to Set Disjointness and Set Intersection, originally due to Kopelowitz, Pettie and Porat~\cite{KopelowitzPP16}: 

\begin{restatable}[Set Disjointness]{problem}{SetDisjointnessProblem}
Given sets $S_1, \dots, S_N \subseteq [U]$ each of size at most $s$ and a set of~$q$ queries $Q \subseteq [N]^2$, report for each query $(i, j) \in Q$ whether $S_i \cap S_j = \emptyset$.
\end{restatable}

\begin{restatable}[Deterministic Set Disjointness Hardness]{theorem}{thmsetdisjointness} \label{thm:set-disjointness}
Let $0 \leq \alpha < 1$. Unless the deterministic 3SUM hypothesis fails, there is no deterministic algorithm for (offline) Set Disjointness with parameters $|U| = \Order(n^{2-2\alpha})$, $N = \Order(n)$, $s = \Order(n^{1-\alpha})$ and $q = \Order(n^{1+\alpha})$ that runs in time $\Order(n^{2-\epsilon})$, for any~\makebox{$\epsilon > 0$}.
\end{restatable}

\begin{restatable}[Set Intersection]{problem}{SetIntersectionProblem}
Given sets $S_1, \dots, S_N \subseteq [U]$ each of size at most $s$ and a set of~$q$ queries $Q \subseteq [N]^2$, report for each query $(i, j) \in Q$ the set intersection $S_i \cap S_j$. Occasionally we specify a size threshold up to which the algorithm is supposed to list elements. 
\end{restatable}

\begin{restatable}[Deterministic Set Intersection Hardness]{theorem}{thmsetintersection} \label{thm:set-intersection}
Let $0 \leq \alpha < 1$ and $0 \leq \beta \leq 1 - \alpha$. Unless the deterministic 3SUM hypothesis fails, there is no deterministic algorithm for (offline) Set Intersection with parameters $|U| = \Order(n^{1+\beta-\alpha})$, \smash{$N = \Order(n^{\frac12 + \frac{\alpha}{2} + \frac{\beta}{2}})$}, $s = \Order(n^{1-\alpha})$ and $q = \Order(n^{1+\alpha})$ that in total lists up to $\Order(n^{2-\beta})$ elements and runs in time $\Order(n^{2-\epsilon})$, for any~\makebox{$\epsilon > 0$}.
\end{restatable}

As announced before, \cref{thm:set-disjointness,thm:set-intersection} immediately entail various deterministic lower bounds for e.g.\ Triangle Listing~\cite{KopelowitzPP16}, which in turn implies deterministic lower bounds e.g.\ for many database problems~\cite{NgoNRR14,KhamisNRR16,BringmannC22,KhamisCKO22}.

We remark that the range of parameters for which our deterministic reduction to Set Intersection applies is limited in comparison to the original randomized reduction. While we can only treat cases where we expect to list at least $\Omega(1)$ elements per queried intersection (on average), Kopelowitz et al.'s reduction also applies to the setting where the number of listed elements is much smaller. For all tight follow-up reductions based on Set Intersection, our parameterization suffices though.

\vskip 1em

\noindent
An answer to Question 1 does not only involve considering selected problems though---there are many scattered 3SUM-based reductions in the literature that have to be considered. We have made the effort to check for all these reductions (that we are aware of) whether our derandomization ideas apply. The number of papers is overwhelming---thus, in order to limit the scope of our project, we have only considered \emph{tight} reductions. In the following subsection we summarize the landscape of 3SUM-based reductions, and summarize which reductions admit derandomizations in what way.

\subsection{The Landscape of 3SUM-Based Reductions}
We have depicted the complete landscape of 3SUM-based reductions (that we are aware of) in \cref{fig:reductions} on \cpageref{fig:reductions}. In this figure, we draw a solid arrow from $P$ to $Q$ if there is a deterministic (tight) reduction from $P$ to $Q$, and a dotted arrow if the reduction is randomized. We have segmented the space of problems into seven regions depending on whether and how derandomizations are known. The light blue regions are deterministic prior to our work, our contribution is that we derandomize all problems in the yellow regions, and the dark blue region remains randomized.

\ifdefined\cameraready
  \input{figures/reductions-camery-ready}
\else
  \input{figures/reductions}
\fi

\subparagraph{(O) Originally Deterministic Reductions.}
Only a minority of 3SUM-based hardness results is deterministic out of the box---most problems rely directly or indirectly on a randomized reduction. This minority includes in particular the numerous geometric problems~\cite{GajentaanO95,BergGO97,ArkinCHMSSY98,BoseKT98,Erickson99,AbellanasHIKLMPS01,BarequetH01,SossEO03,ArchambaultEK05,EricksonHM06,CheongEH07,AronovH08}. These reductions usually build on simple algebraic transformations of the input. For instance, the reduction from 3SUM to the problem of finding a triple of colinear points simply creates point $(a, a^3)$ for each input number~\makebox{$a \in A$}. As a consequence all these reductions are deterministic.

Other two examples of similarly simple deterministic constructions are reductions to Hamming Pattern Matching under Polynomial Transformation~\cite{ButmanCCJLPPS13}, and to fully retroactive 3SUM data structures~\cite{ChenDGWXY18}.

\subparagraph{(C) Reductions via Convolution 3SUM\ifdefined\cameraready.\fi}
As mentioned before, the next wave of 3SUM-hard problems often involved reductions via the Convolution 3SUM problem. Except for the randomization in the reduction to Convolution 3SUM, many of these reductions were deterministic, and were thus fully derandomized by~\cite{ChanH20}. This list of problems includes Zero Triangle~\cite{WilliamsW13}, Jumbled Indexing~\cite{AmirCLL14}, Subset Sum for $k$-Enclosing Rectangle~\cite{ChanH19}, and Hausdorff Distance under translation~\cite{BringmannN21b}. The Zero Triangle problem was then deterministically reduced to the Matching Triangle and Triangle Collection problems~\cite{AbboudWY18}, which in turn led to deterministic reductions to a number of dynamic problems~\cite{AbboudWY18,Dahlgaard16}.

\subparagraph{(S) Reductions Based on the Self-Reduction.}
Another class of reductions that was already deranomized prior to our work is reductions that rely on an efficient \emph{self-reduction} for 3SUM (i.e., a reduction that reduces 3SUM on $n$ numbers to $n^{2\alpha}$ instances on $n^{1-\alpha}$ numbers each, for any $\alpha \in [0, 1]$). The first 3SUM self-reduction was given implicitly by Patrascu~\cite{Patrascu10}, and it was randomized. Later, deterministic self-reductions were proposed~\cite{LincolnWWW16, GronlundP18}. Based on these deterministic self-reductions, one can immediately derandomize the reductions to All-Numbers 3SUM~\cite{WilliamsW18}, and to detecting a 3-star or a 3-matching of total weight zero in edge-weighted graphs~\cite{AbboudL13}.

\subparagraph{(D) Reductions via Set Disjointness, Set Intersection, and Triangle Problems.}
Another major class of reductions, as announced before, is via the Set Disjointness or Set Intersection problems. There are deterministic reductions from Set Intersection to Triangle Listing (in various parameterizations, also involving the arboricity of the given graph)~\cite{KopelowitzPP16}, and via Triangle Listing to Dictionary Matching with one gap~\cite{AmirKLPPS19} and exists-connectivity queries in graph timelines~\cite{KarczmarzL15}. Besides, Triangle Listing reduces to several database problems~\cite{NgoNRR14,KhamisNRR16,BringmannC22,KhamisCKO22}. The Set Disjointness problem reduces trivially to the All-Edges Triangle problem which in turn reduces to several range-query problems~\cite{DurajK0W20}. As a consequence of our \cref{thm:set-disjointness,thm:set-intersection}, we have successfully derandomized all these reductions.

\subparagraph{(U) Reductions Relying on Small Universe Size.}
Other reductions are not inherently randomized, but rather rely on a randomized universe reduction. We can derandomize all such reductions by replacing the randomized universe reduction with our deterministic one (\cref{thm:univ-reduction}). Reductions of this type have been established for 3-Linear Degeneracy Testing~\cite{DGS20}\footnote{In \cite{DGS20} the authors define 3SUM and the class of 3-Linear Degeneracy Testing problems over a universe of cubic size $n^3$, and thereby implicitly rely on a universe reduction.} and grammar-compressed variants of vector inner product and matrix-vector multiplication~\cite{AbboudBBK20}. Similarly, the Dynamic $k$-Mismatch problem~\cite{CliffordGK0U22} relies on a quadratic-size universe reduction for All-Numbers 3SUM, and the Odd Abelian Square problem~\cite{RadoszewskiRSWZ21} relies on a quadratic-size universe reduction for Convolution 3SUM. To additionally derandomize these problems we extend our universe reduction for 3SUM to also cover these two cases (\cref{lem:univ-reduction-an-3sum,lem:univ-reduction-conv-3sum}).

\subparagraph{(A) Reductions via Ad-Hoc Derandomizations.}
Fortunately, most hardness results have been derandomized using our general toolset. There are however a couple of reductions that required additional tailored arguments on top of that. The randomized reduction to Monochromatic Convolution~\cite{LincolnPW20} relies, among other things, on a trick of adding random offsets to merge several sparse sets into a dense one without (too many) collisions. We provide a deterministic variant of this trick, and successfully derandomize the reduction. Next, to derandomize the reduction to Local Alignment~\cite{AbboudWW14} we follow the approach of Chan an He~\cite{ChanH20} of controlling the number of collisions, and add the trick of Abboud, Lewi, and Williams~\cite{AbboudLW14} of replacing numbers with multidimensional vectors with small entries. Finally, the reduction to Convolution Witnesses~\cite{GoldsteinKLP16} can be derandomized by plugging in our method for selecting a hash function with few false positives.

\subparagraph{(R) Reductions that Remain Randomized.}
Unfortunately the list of known 3SUM-based reductions includes three more results which we did not manage to derandomize~\cite{AbboudBKZ22,AbboudBF23,JinX23}. These papers rely on the recent ``short cycle removal'' technique that was originally proposed in~\cite{AbboudBKZ22} and optimized in~\cite{AbboudBF23,JinX23}, and led to tight lower bounds for 4-Cycle Listing, Approximate Distance Oracles with stretch close to~$2$ or~$3$, and for $4$-Linear Degeneracy Testing. At the heart of short cycle removal lies a structure-versus-randomness dichotomy which involves, in~\cite{AbboudBF23,JinX23}, finding \emph{structured} subsets via the algorithmic Balog-Szemerédi-Gowers (BSG) theorem~\cite{ChanL15}. Derandomizing the BSG algorithm constitutes the major challenge in derandomizing these reductions (though not the only one---the reductions also rely on a specialized kind of additive hashing which we have not attempted to derandomize here). We leave it as an important open problem to derandomize these few remaining reductions.


\subsection{Outline}
In \cref{sec:preliminaries} we fix some preliminaries. In \cref{sec:toolkit} we introduce our derandomization toolkit---this section starts with a high-level technical overview before diving into the proofs. In \cref{sec:univ-reduction} we combine our tools to obtain the deterministic universe reduction, and in \cref{sec:set-disjointness} we provide the deterministic hardness results for Set Disjointness and Set Intersection.
\ifdefined\FullVersionArXiv
  In \cref{sec:mono-conv,sec:local-alignment,sec:conv-witness}
\else
  In the full version of the paper
\fi
we provide even more derandomizations: for the Monochromatic Convolution, Local Alignment and Convolution Witness problems. 
\section{Preliminaries} \label{sec:preliminaries}
We use the standard notation $[n] := \set{1, \dots, n}$ and $\widetilde\Order(T) = T (\log T)^{\Order(1)}$. Throughout this paper, we consider the 3SUM problem (and its variants) formally defined as follows:




\begin{problem}[3SUM]
Decide whether in a given set $A \subseteq [U]$ there is a triple $a, b, c \in A$ with $a + b = c$.
\end{problem}

\begin{problem}[All-Numbers 3SUM]
Given a set $A \subseteq [U]$, decide for each $c \in A$ whether there is a pair $a, b \in A$ with $a + b = c$.
\end{problem}

\begin{problem}[Convolution 3SUM]
Given a vector $X \in [U]^n$, decide for each $k \in [n]$ whether there is a pair $i, j \in [n]$ with $i + j = k$ and $X[i] + X[j] = X[k]$.
\end{problem}

While we have stated the problems in a \emph{monochromatic} form (i.e., for a single set $A$ where we select $a, b, c \in A$), it is well-known and easy to check that they are equivalent to their \emph{trichromatic} variants of these problems (i.e., for three sets $A, B, C$ where we select $a \in A,\, b \in B,\, c \in C$) (by deterministic reductions). We occasionally rely on this statement in our proofs.
\section{Our Toolkit} \label{sec:toolkit}
In this section we develop the toolkit that is necessary to prove our main derandomizations. We will start with a technical overview of the two tools in \cref{sec:toolkit:sec:hashing-overview,sec:toolkit:sec:self-reduction-overview}, and provide the missing formal proofs in \cref{sec:toolkit:sec:proofs}.

\subsection{Tool 1: Deterministic Additive Hashing} \label{sec:toolkit:sec:hashing-overview}
A hash function $h$ is \emph{additive} if there is some modulus $m$, such that $h(a) + h(b) = h(a + b) \bmod m$ for all inputs $a, b$. Many 3SUM-hardness reductions involve additive hash functions (or \emph{pseudo-additive} functions, which have the slightly weaker property that $h(a) + h(b) - h(a + b)$ takes only a constant number of values), and one of the simplest randomized constructions is the family of functions~\makebox{$h(x) = x \bmod m$} where $m$ is a \emph{random} prime of prescribed size.

In our setting we are required to select an additive hash function~$h$ \emph{deterministically}. Previous work~\cite{ChanL15,ChanH20} has already faced the same challenge, and provided the following simple solution. The insight is that function $h(x) = x \bmod m$ is still an effective hash function even if $m$ is the product of several smaller primes $m = p_1 \cdot \ldots \cdot p_R$, $p_i \approx m^{1/R}$. In order to select $m$, one can follow the method of conditional expectations: Instead of picking all primes at the same time, we will fix~\makebox{$p_1, \dots, p_R$} step by step. In each step, we pick a prime $p_i$ which is \emph{at least as good as what we would expect from a random prime.} In particular, it suffices to pick a locally optimal prime. To this end, we can exhaustively enumerate all primes of size~\makebox{$\approx m^{1/R}$} and test which one suits us best. All in all, this approach works to deterministically construct hash functions for all properties that can be efficiently tested.

\subparagraph{Additive Hashing with Few Collisions.}
One such example is Chan and He's deterministic reduction from 3SUM to Convolution 3SUM~\cite{ChanH20}. Specifically, they require a hash function which has few \emph{collisions} (i.e., pairs~\makebox{$a, b \in A$} with $h(a) = h(b)$). Since it is easy to count the number of collisions for a given hash function in linear time, the above recipe leads to the following derandomization:

\begin{lemma}[Deterministic Additive Hashing with Few Collisions~\cite{ChanL15}] \label{lem:hashing-few-collisions}
Let $0 \leq \mu \leq 2, \delta > 0$. There is a deterministic algorithm that, given a set $A \subseteq [U]$ of size $n$, finds a modulus $m \in [n^\mu, 2n^\mu)$ such that
\begin{equation*}
    \#\set{(a, b) \in A^2 : a \equiv b \mod m} \leq n^{2-\mu} (\log U)^{\Order(1/\delta)}.
\end{equation*}
The algorithm runs in time $\Order(n^{1+\delta})$.
\end{lemma}

\subparagraph{Additive Hashing with Few False Positives.}
Unfortunately, for most 3SUM-based reductions, bounding the number of collisions is not sufficient though. The property that most reductions rely on is that the number of triples $a, b, c \in A$ with $a + b \equiv c \mod m$ is small---we often refer to such a triple as a \emph{pseudo-solution} or a \emph{false positive}. Unfortunately, counting the number of false positives is in general a hard problem---it is in fact another 3SUM instance $\set{a \bmod m : a \in A}$ which cannot be expected to be solvable in subquadratic time.

To circumvent this barrier, we exploit a simple observation: Note that the universe size of the reduced 3SUM instance is only $m$, hence we can count the 3SUM solutions in time $\widetilde\Order(m)$ using the Fast Fourier Transform. Building on that insight, we establish the following new lemma.

\begin{lemma}[Deterministic Additive Hashing with Few False Positives] \label{lem:hashing-few-pseudo-solutions}
Let $0 \leq \mu < 3, \delta > 0$. There is a deterministic algorithm that, given 3SUM instances $A_1, \dots, A_g \subseteq [U]$ of size $\Order(n)$, either finds a positive modulus $m \in [n^\mu, 2n^\mu)$ such that
\begin{equation*}
    \sum_{i=1}^g \#\set{(a,b,c) \in A_i^3 : a + b \equiv c \mod m } \leq g n^{3-\mu+\delta} (\log U)^{\Order(1/\delta)},
\end{equation*}
or decides that at least one instance is a yes-instance. The algorithm runs in $\widetilde\Order(g n^{\max(\mu, 1)+\delta})$ time.
\end{lemma}

We remark that while we state this lemma in a slightly more general form, in almost all applications we only need the lemma for the special case $g = 1$.

\subsection{Tool 2: Deterministic Self-Reduction} \label{sec:toolkit:sec:self-reduction-overview}
As our second major tool we rely on a deterministic \emph{self-reduction} for 3SUM. Here by a self-reduction we mean an algorithm that reduces a 3SUM instance $A$ to some other 3SUM instances~\makebox{$A_1, \dots, A_N$} with the property that~$A$ is a yes-instance if and only if there is a yes-instance $A_i$. This reduction is considered efficient only if $\sum_i |A_i|^2 \leq \widetilde\Order(|A|^2)$ (i.e., if the brute-force running times match).

Self-reductions for 3SUM are known based on two very different approaches. The first approach relies on additive hashing. The second approach, originally due to~\cite{LincolnWWW16,GronlundP18} (with ideas borrowed from~\cite{CzumajL09}), is conceptually simpler and also deterministic. The idea is to bucket $A$ into smaller sets in such a way that only few triples of buckets can contain a solution. For our purposes we need a slightly stronger version than what was stated in the previous papers, where we also consider the universe size of the constructed subinstances:

\begin{lemma}[Deterministic Self-Reduction] \label{lem:dominance}
Let $A \subseteq [U]$ be a set of size $n$, and let $g \geq 1$. We can deterministically construct a partition of $A$ into subsets $A_1, \dots, A_g$, and a set of triples $R \subseteq [g]^3$ of size~\makebox{$\Order(g^2)$} such that:
\begin{itemize}
    \item Each set $A_i$ has size $\Order(n / g)$ and can be covered by an interval of length $\Order(U / g)$.
    \item For all triples $a, b, c \in A$ with $a + b = c$, there is some $(i, j, k) \in R$ with $a \in A_i,\, b \in A_j,\, c \in A_k$.
\end{itemize}
The algorithm runs in time $\widetilde\Order(n + g^2)$.
\end{lemma}

\subsection{Proofs} \label{sec:toolkit:sec:proofs}
We finally provide the missing proofs of \cref{lem:hashing-few-pseudo-solutions,lem:dominance}.

\begin{proof}[Proof of \cref{lem:hashing-few-pseudo-solutions}]
For a modulus $m$, we call a triple $(a, b, c) \in A_i^3$ satisfying $a + b \equiv c \mod m$ a \emph{pseudo-solution} of $m$. Let us write
\begin{equation*}
    S(m) = \sum_{i=1}^g \#\set{(a,b,c) \in A_i^3 : a + b \equiv c \mod m }
\end{equation*}
to denote the number of pseudo-solutions of $m$; our goal is to compute a modulus $m$ that with few pseudo-solutions, $S(m) \leq g n^{3-\mu} (\log U)^{\Order(1/\delta)}$. As outlined before, our strategy is to let $m$ be the product of several small primes, each of size roughly $n^\delta$. These primes are selected one by one, such that each next prime narrows down the number of false positives by an appropriate amount. Specifically, consider the following process: We initialize $m \gets 1$ and $P \subseteq [n^\delta, 2n^\delta)$ to be the subset of primes of size roughly~$n^\delta$, and run the following three steps:
\begin{enumerate}
    \item For each prime $p \in P$, compute $S(m \cdot p)$ using FFT.
    \item Select the prime $p \in P$ that minimizes $S(m \cdot p)$. Update $m \gets m \cdot p$ and $P \gets P \setminus \set{p}$.
    \item If $m \geq \frac12 n^{\mu - \delta}$, then update $m \gets m \cdot \ceil{n^\mu / m}$ and stop. Otherwise, go to step 1.
\end{enumerate}

\subparagraph{Running Time.}
The easier part of the proof is to bound the running time of this process. Note that we increase~$m$ by a factor of at least $n^\delta$ in each step. Hence, after at most $\Order(1)$ iterations we have increased~$m$ beyond the threshold $\frac12 n^{\mu - \delta}$ and the process stops. The running time is dominated by step 1: For each prime $p$, we compute $S(m \cdot p)$ using the Fast Fourier Transform. More precisely, for each 3SUM instance~$A_i$ we first prepare the (multi-)set~\makebox{$A_i \bmod {(mp)}$} obtained by hashing all numbers modulo $mp$. We can compute the number of 3SUM solutions in this reduced instance in time $\Order(mp \log (mp)) = \widetilde\Order(n + n^\mu)$ (using that $m \leq n^{\mu - \delta}$ and $p \leq n^\delta$). Repeating this call for all~$g$ instances and all $\Order(n^\delta)$ primes amounts for time~\smash{$\widetilde\Order(n^\delta \cdot g \cdot (n + n^\mu)) = \widetilde\Order(g n^{\max(1, \mu) + \delta})$}.

\subparagraph{Correctness.}
It remains to bound $S(m)$, where is $m$ is the modulus computed by the above process. Let~$m_i$ denote the modulus at the $i$-th step of the process (with $m_0 = 1$) and assume that all given 3SUM instances are unsatisfiable. We prove by induction that
\begin{equation*}
    S(m_i) \leq \frac{g n^3}{n^{i \delta}} \cdot (2 \log U)^{i}.
\end{equation*}
For $i = 0$, this is clearly true. So let $i > 0$ and assume that the induction hypothesis holds for~\makebox{$i - 1$}. Our strategy is to show that for a \emph{random} prime $p \in [n^\delta, 2n^{\delta})$, with good probability we have
\begin{equation*}
    S(m_{i-1} \cdot p) \leq \frac{g n^3}{n^{i \delta}} \cdot \Order(\log U)^{i}.
\end{equation*}
Consider an arbitrary triple~\makebox{$a, b, c \in [U]$} with $a + b \neq c$. The event $a + b \equiv c \bmod p$ is equivalent to~$p$ being a divisor of the number $a + b - c \in [-U\,..\,2U]$, which happens with probability at most
\begin{equation*}
    \frac{\log_{n^\delta}(2U)}{\Omega(n^\delta \log^{-1}(n))} = \Order(n^{-\delta} \log U),
\end{equation*}
since $\log_{n^\delta}(2U)$ is the maximum number of divisors of size at least $n^\delta$ of $a + b - c$, and since there are~\smash{$\Omega(n^\delta \log^{-1}(n))$} primes in $[n^\delta, 2n^\delta)$ by the Prime Number Theorem. Therefore, among the pseudo-solutions modulo $m_{i-1}$, only an expected $\Order(n^{-\delta} \log U)$-fraction is also a pseudo-solution modulo $m_{i-1} \cdot p$. Therefore, by Markov's inequality with probability at least~$\frac12$ we have that
\begin{align*}
    S(m_{i-1} \cdot p)
    &\leq 2 \cdot S(m_{i-1}) \cdot \Order(n^{-\delta} \log U) \\
    &\leq \frac{g n^3}{n^{(i-1) \delta}} \cdot \Order(\log U)^{i-1} \cdot \Order(n^{-\delta} \log U) \\
    &= \frac{g n^3}{n^{i \delta}} \cdot \Order(\log U)^{i}.
\end{align*}
It follows that there is a successful prime $p \in [n^\delta, 2n^{\delta})$ with \smash{$S(m_{i-1} \cdot p) \leq \frac{g n^3}{n^{i \delta}} \cdot \Order(\log U)^{i}$}. Recall that the algorithm selects the prime that \emph{minimizes} the number of pseudo-solutions $S(m_{i-1} \cdot p)$ and so indeed we have \smash{$S(m_i) \leq  \frac{g n^3}{n^{i \delta}} \cdot \Order(\log U)^{i}$}.

Finally, recall that the process terminates no sooner than $m_i \geq \frac12 n^{\mu - \delta}$. Since $n^\delta \leq m_i < (2n^\delta)^i$, it follows that the final iteration count is at least $i \geq \frac{\mu}{\delta} - \Order(1)$ and at most $i \leq \frac\mu\delta$. Hence, for the final modulus $m$ we have indeed
\begin{equation*}
    S(m) \leq \frac{gn^3}{n^{\mu - \Order(\delta)}} \cdot \Order(\log U)^{\frac\mu\delta} = g n^{3-\mu+\Order(\delta)} \cdot (\log U)^{\Order(1/\delta)}.
\end{equation*}
By decreasing $\delta$ by a constant factor, we obtain the claimed bound. Recall that this bound is conditioned on the assumption that the given 3SUM instances are no-instances. However, if our algorithm computes a modulus $m$ with unexpectedly many pseudo-solutions (which we can verify in time $\widetilde\Order(g n^{\max(1, \mu) + \delta})$), we can infer that at least one of the given instances is a yes-instance.

As a final detail, we show that the algorithm produces a modulus $m$ in the range $[n^\mu, 2n^{\mu})$. Indeed, in line 3, we execute the final update $m \gets m \cdot \ceil{n^\mu / m}$ only after we have increased~$m$ to be in the range $\frac12 n^{\mu-\delta} \leq m < n^\mu$. It follows that then $m \cdot \ceil{n^\mu / m} \leq m \cdot (n^\mu / m + 1) < 2n^\mu$.
\end{proof}

This completes the treatment of the deterministic additive hashing tool. Let us next prove the refined deterministic self-reduction.

\begin{proof}[Proof of \cref{lem:dominance}]
Our goal is to partition $A$ into subsets~$A_1, \dots, A_g$ with the following two properties:
\begin{itemize}
    \item $|A_i| \leq 2n / g$, and
    \item $\max(A_i) - \min(A_i) \leq 2 U / g$ (i.e., $A_i$ can be covered by an interval of length $2 U / g$).
\end{itemize}
By constructing this partition greedily, by scanning over the elements in $A$ in sorted order including all elements until one of the two conditions is violated, it is easy to see that both rules lead to at most $g / 2$ stops and therefore $g$ groups suffice in total. We assume that the groups constructed in this way are ordered in the natural way (i.e., $\max(A_i) < \min(A_{i+1})$ for all $i$).

The next step is to construct the set $R \subseteq [g]^3$. Naively, including all~$g^3$ triples $(i, j, k)$ would be correct, if there is a 3-sum in the original instance it is certainly contained in one of the subinstance~$(A_i, A_j, A_k)$. However, many of these instances are \emph{trivial} in the sense that either
\begin{enumerate}
    \item $\min(A_i) + \min(A_j) > \max(A_k)$, or
    \item $\max(A_i) + \max(A_j) < \min(A_k)$;
\end{enumerate}
both conditions make it impossible for $(A_i, A_j, A_k)$ to contain a 3-sum. To improve the reduction, we will therefore include in $R$ only triples that are not trivial.

We argue that the number of remaining instances (that is, the number of instances that falsify both previous conditions) is at most~$\Order(g^2)$. To see this, consider the partially ordered set~\makebox{$([g]^3, \prec)$} where we let~\makebox{$(i, j, k) \prec (i', j', k')$} if and only if~\makebox{$\max(A_i) < \min(A_{i'})$}, \makebox{$\max(A_j) < \min(A_{j'})$} and \makebox{$\max(A_k) < \min(A_{k'})$}. With this definition, if $(i, j, k) \prec (i', j', k')$, then~$(i, j, k)$ or $(i', j', k')$ must be a trivial no-instance. Indeed, suppose that $(i, j, k)$ and $(i', j', k')$ both falsify the conditions~(1) and~(2). Then:
\begin{align*}
    0 &\leq \max(A_i) + \max(A_j) - \min(A_k) &&\text{(since $(i, j, k)$ falsifies (2))} \\
    &< \min(A_{i'}) + \min(A_{j'}) - \max(A_{k'}) &&\text{(since $(i, j, k) \prec (i', j', k')$)} \\
    &\leq 0. &&\text{(since $(i', j', k')$ falsifies (1))}
\end{align*}
More generally, on any chain in the lattice there can be at most one nontrivial instance. Therefore, in order to bound the number of nontrivial instances, it suffices to cover $[g]^3$ by $\Order(g^2)$ chains. This cover is easy to construct: For all~\makebox{$a, b \in \set{-g, \dots, g}$}, take the chain~\makebox{$\set{(i, i + a, -i + b) : i \in [g]}$}. Any triple $(i, j, k)$ is covered by the chain for $a = j - i$ and~\makebox{$b = k + i$}.

It remains to analyze the running time. Construction the partition of $A$ takes linear time after sorting~$A$ in time $\Order(n \log n)$. We can construct the set $R$ by enumerating all $(i, j) \in [g]^2$, and enumerating for each such pair only the subset of $k$'s with $\min(A_i) + \min(A_j) \leq \max(A_k)$ and $\max(A_i) + \max(A_j) \geq \min(A_k)$. By preprocessing~$A$ with a range query data structure in near-linear time, this step runs in time~\smash{$\widetilde\Order(g^2 + |R|) = \widetilde\Order(g^2)$}.
\end{proof}

This lemma immediately implies the following deterministic self-reductions for 3SUM and the all-numbers variant of 3SUM. (Here, the running time increases to $\widetilde\Order(n g)$ since we explicitly write down all constructed instances, instead of concisely describing the instances via subsets of the original instance.)

\begin{corollary}[Deterministic Self-Reduction for 3SUM] \label{lem:self-reduction-3sum}
For any $g \geq 1$, a given 3SUM instance of size~$n$ over the universe $[U]$ can deterministically be reduced to $\Order(g^2)$ 3SUM instances of size~$\Order(n / g)$ over the universe~\makebox{$[\Order(U / g)]$}. The running time of the reduction is \smash{$\widetilde\Order(n g)$}.
\end{corollary}

\begin{corollary}[Deterministic Self-Reduction for All-Numbers 3SUM] \label{lem:self-reduction-an-3sum}
For any $g \geq 1$, a given AN-3SUM instance of size $n$ over the universe $[U]$ can deterministically be reduced to $\Order(g^2)$ AN-3SUM instances of size~$\Order(n / g)$ over the universe~\makebox{$[\Order(U / g)]$}. The running time of the reduction is~\smash{$\widetilde\Order(n g)$}.
\end{corollary}
\section{Deterministic Universe Reduction for 3SUM} \label{sec:univ-reduction}
In this section we combine the tools established in \cref{sec:toolkit} to deduce a deterministic universe reduction for 3SUM (\cref{thm:univ-reduction}). We later extend the universe reduction also to All-Numbers 3SUM and Convolution 3SUM (\cref{lem:univ-reduction-an-3sum,lem:univ-reduction-conv-3sum}). We start with a trivial reduction which we need in the proof of \cref{thm:univ-reduction}.

\begin{observation}[Trivial Universe Reduction for 3SUM] \label{obs:univ-reduction-trivial}
Let $U \geq U'$. A 3SUM instance over~$[U]$ can be reduced to $\Order((\frac{U}{U'})^3)$ many 3SUM instances over $[U']$ (of the same size $n$) such that the solutions are in 1-to-1 correspondence. The reduction runs in time $\Order(n \cdot (\frac{U}{U'})^3)$.
\end{observation}
\begin{proof}
Chop the given set $A$ into $\ceil{\frac{U}{U'}}$ subsets that can be covered by an interval of length~$U'$. In this way we construct at most $\ceil{\frac{U}{U'}}^3$ triples each of which can be viewed as a 3SUM instance over the universe $[U']$. (In fact the number of relevant 3SUM solutions is only $\Order((\frac{U}{U'})^2)$, but we do not need this improvement here.) 
\end{proof}

\thmunivreduction*
\begin{proof}
Assume that 3SUM over the universe $[n^3]$ can be solved in time $\Order(n^{2-\epsilon})$ for some $\epsilon > 0$. We are designing a subquadratic algorithm for a given 3SUM instance $A \subseteq [U]$. Our reduction runs in four steps:
\begin{enumerate}
    \item Let $\mu, \delta > 0$ be parameters to be determined later. We use the deterministic additive hashing lemma (\cref{lem:hashing-few-pseudo-solutions} with parameters~$\mu$,~$\delta$ and $g = 1$) to find a modulus~$m$ with~$m = \Theta(n^{\mu})$ and
    \begin{equation*}
        S := \#\set{(a, b, c) \in A^3 : a + b \equiv c \mod m} \leq n^{3-\mu+\delta} (\log U)^{\Order(1/\delta)}.
    \end{equation*}
    Alternatively, if \cref{lem:hashing-few-pseudo-solutions} reports that the instance contains a solution, we can immediately stop here. As before let us refer to triples $(a, b, c)$ with $a + b \equiv c \mod m$ as \emph{pseudo-solutions.} Our strategy is to design a subquadratic algorithm that lists all pseudo-solutions. Afterwards, it is easy to check if a genuine solution is among them. Note that the pseudo-solutions stand in 1\=/to\=/$\Order(1)$-correspondence to the solutions of the 3SUM instance
    \begin{equation*}
        A' := \set{a \bmod m, (a \bmod m) + m : a \in A},
    \end{equation*}
    hence in the following steps it suffices to list all solutions in $A'$. 
    \item Let $\alpha$ be a parameter to be determined. We apply the deterministic self-reduction (\cref{lem:self-reduction-3sum} with parameter $g = n^\alpha$) to $A'$ to construct $\Order(n^{2\alpha})$ 3SUM instances $A_i'$, each of size $\Order(n^{1-\alpha})$ and over a universe of size~$\Order(m / n^\alpha) = \Order(n^{\mu-\alpha})$.
    \item Using the trivial universe reduction (\cref{obs:univ-reduction-trivial} with parameters $U = \Theta(n^{\mu-\alpha})$ and $U' = \Theta(n^{3-3\alpha})$) we can further reduce each instance $A_i'$ to \smash{$\max(1, (\frac{n^{\mu-\alpha}}{n^{3-3\alpha}})^3)$} 3SUM instances of the same size $\Order(n^{1-\alpha})$ and over a universe of cubic size $\Order(n^{3-3\alpha})$. We solve each such instance using the fast algorithm in time $\Order(n^{(1-\alpha)(2-\epsilon)})$. (Here we only detect whether $A_i'$ contains a solution, but do not list all solutions.)
    \item For each set $A_i'$ for which we have detected a solution in the previous step, we exhaustively enumerate all solutions in time $\Order(|A_i'|^2) = \Order(n^{2-2\alpha})$.
\end{enumerate}

\subparagraph{Running Time.}
The correctness is clear, but it remains to analyze the running time. In what follows we choose the parameters $\delta = \frac{\epsilon}{32} > 0,\, \mu = 2 - 2\delta,\, \alpha = \frac12 + 2\delta$ and analyze the contributions step-by-step:
\begin{enumerate}
    \item Running \cref{lem:hashing-few-pseudo-solutions} takes time $\Order(n^{\max(1, \mu) + \delta}) = \Order(n^{2-\delta})$.
    \item Running \cref{lem:self-reduction-3sum} takes time $\widetilde\Order(n g) = \widetilde\Order(n^{3/2+2\delta})$.
    \item The overhead due to \cref{obs:univ-reduction-trivial} is negligible, and this step is dominated by solving the $\Order(n^{2\alpha} \cdot \max(1, \frac{n^{\mu-\alpha}}{n^{3-3\alpha}})^3)$ instances using the fast algorithm. Each call takes time $\Order(n^{(1-\alpha)(2-\epsilon)})$, hence in total this step takes time
    \begin{gather*}
        \Order(n^{2\alpha} \cdot \max(1, \tfrac{n^{\mu-\alpha}}{n^{3-3\alpha}})^3 \cdot n^{(1-\alpha)(2-\epsilon)}) \\
        \qquad= \Order(n^{2 - \epsilon(1 - \alpha)} \cdot \max(1, \tfrac{n^{3/2-4\delta}}{n^{3/2-6\delta}})^3) \\
        \qquad= \Order(n^{2 - \epsilon(1 - \alpha) + 6\delta}) \\
        \qquad= \Order(n^{2-2\delta}),
    \end{gather*}
    where in the last step we used that $1 - \alpha \geq \frac14$.
    \item Recall that we only brute-force 3SUM solutions $A_i'$ which contain a solution. Moreover, since the solutions across the sets $A_i'$ are in $\Order(1)$-to-1 correspondence to the at most $S$ pseudo-solutions, we brute-force at most $\Order(S)$ instances in total time
    \begin{equation*}
        \Order(S \cdot n^{2-2\alpha}) =  \Order(n^{3-\mu+\delta + 2-2\alpha} (\log U)^{\Order(1/\delta)}) = \Order(n^{2-\delta} (\log U)^{\Order(1/\delta)}).
    \end{equation*}
\end{enumerate}
Summing over all contributions, the total running time is bounded by $n^{2-\delta} (\log U)^{\Order(1/\delta)}$, which is as claimed.
\end{proof}

\begin{corollary}[3SUM over Cubic Universes]
For any $\epsilon > 0$, the 3SUM problem over a universe of size $n^3$ cannot be solved in deterministic time $\Order(n^{2-\epsilon})$, unless the deterministic 3SUM hypothesis fails.
\end{corollary}

\subsection{Universe Reduction for All-Numbers 3SUM}
For the All-Numbers 3SUM problem there exists an even better randomized universe reduction---to a universe of quadratic size $n^2$. To derandomize this universe reduction as well, we follow the same ideas as in \cref{thm:univ-reduction}:

\begin{lemma}[All-Numbers 3SUM over Quadratic Universes] \label{lem:univ-reduction-an-3sum}
For any $\epsilon > 0$, the AN-3SUM problem over a universe of size $n^2$ cannot be solved in deterministic time $\Order(n^{2-\epsilon})$, unless the deterministic 3SUM hypothesis fails.
\end{lemma}
\begin{proof}
Suppose that AN-3SUM over quadratic universes can be solved in time $\Order(n^{2-\epsilon})$. We design a 3SUM algorithm for $A \subseteq [n^c]$ in time $\Order(n^{2-\epsilon'})$. The reduction again runs in four steps:
\begin{enumerate}
    \item We first apply \cref{lem:hashing-few-pseudo-solutions} (with parameters $\mu, \delta > 0$ to be determined later and $g = 1$) to find a modulus $m = \Theta(n^\mu)$ such that
    \begin{equation*}
        S := \#\set{(a, b, c) \in A^3 : a + b \equiv c \mod m} \leq n^{3-\mu+\delta} (\log n)^{\Order(1/\delta)} = \widetilde\Order(n^{3-\mu+\delta}).
    \end{equation*}
    Let us again refer to all triples $(a, b, c)$ with $a + b \equiv c \mod m$ as \emph{pseudo-solutions.} Then the pseudo-solutions of the original instance are in 1-to-$\Order(1)$ correspondence to the solutions of the 3SUM instance
    \begin{align*}
        A' &= \set{a \bmod m, (a \bmod m) + m : a \in A}.
    \end{align*}
    \item We apply \cref{lem:dominance} with parameter $g = n^\alpha$ to partition $A'$ into subsets $A_1', \dots, A_g'$ such that each part has size at most $\Order(n^{1-\alpha})$, with only $\Order(n^{2\alpha})$ triples of parts that may contain a 3-sum. (This reduction even reduces the universes of the subinstances by a factor $n^\alpha$, but we don't need this improvement here).
    \item We apply the trivial universe reduction (\cref{obs:univ-reduction-trivial} with parameters $U = \Theta(n^\mu)$ and $U' = \Theta(n^{2-2\alpha})$) to further reduce each instance $A_i'$ to $\max(1, \frac{n^\mu}{n^{2-2\alpha}})^3$ instances of the same size over a universe of quadratic size $\Order(n^{2-2\alpha})$. We solve each new instance using the fast AN-3SUM algorithm in total time $\Order(n^{2\alpha} \cdot \max(1, \frac{n^\mu}{n^{2-2\alpha}})^3 \cdot n^{(1-\alpha)(2-\epsilon)})$.
    \item For each relevant triple $(i, j, k)$, and for each $a \in A_i'$ that was found to be part of a pseudo-solution, we explicitly list all pseudo-solutions in $A_i' \cup A_j' \cup A_k'$. As this step lists all pseudo-solutions, we are guaranteed to find any proper 3-sum in this step.
\end{enumerate}

\subparagraph{Running Time.}
It is easy to see that this algorithm is correct, but it remains to analyze the running time. We choose the parameters \smash{$\delta := \frac{\epsilon}{16}$}, $\mu := 2 - 2\delta$, and $\alpha = 4\delta$, and analyze the running time step by step:
\begin{enumerate}
    \item Running \cref{lem:hashing-few-pseudo-solutions} takes time $\Order(n^{\max(\mu, 1) + \delta}) = \Order(n^{2-\delta})$.
    \item Running \cref{lem:dominance} takes time $\widetilde\Order(n + g^2) = \widetilde\Order(n)$.
    \item The trivial universe reduction is negligible, hence this step runs in total time
    \begin{equation*}
        \Order(\max(1, \tfrac{n^\mu}{n^{2-2\alpha}})^3 \cdot n^{(1-\alpha)(2-\epsilon)}) = \Order(n^{2-(1-\alpha)\epsilon + 6\delta}) = \Order(n^{2-2\delta}),
    \end{equation*}
    using in the last step that $1 - \alpha \geq \frac12$.
    \item We brute-force $S = \widetilde\Order(n^{3-\mu+\delta})$ many instances of size $\Order(n^{1-\alpha})$, taking time $\widetilde\Order(n^{2-\delta})$.
\end{enumerate}
Summing over all contributions we obtain that the total running time is $\Order(n^{2-\delta})$ as claimed.
\end{proof}

\subsection{Universe Reduction for Convolution 3SUM}
Finally we also prove that by means of deterministic reductions, the Convolution 3SUM problem is already hard over quadratic universes. This proof is essentially identical to~\cite{ChanH20}, except that we start with our deterministic universe reduction for 3SUM.

\begin{lemma}[Convolution 3SUM over Quadratic Universes] \label{lem:univ-reduction-conv-3sum}
For any $\epsilon > 0$, the Convolution 3SUM problem over a universe of size $n^2$ cannot be solved in deterministic time $\Order(n^{2-\epsilon})$, unless the deterministic 3SUM hypothesis fails.
\end{lemma}
\begin{proof}
Suppose that the Convolution 3SUM problem over a universe of size $n^2$ can be solved in deterministic time $\Order(n^{2-\epsilon})$ for some $\epsilon > 0$. In order to falsify the deterministic 3SUM hypothesis, it suffices to design a subquadratic algorithm for a given 3SUM instance $A \subseteq [n^3]$. Using the deterministic additive hashing construction with few \emph{collisions} (\cref{lem:hashing-few-collisions} with parameters $\mu = 1$ and $\delta = \frac12$, say), we can find a modulus~\makebox{$m = \Theta(n)$} such that
\begin{equation*}
    \#\set{(a, b) \in A^2 : a \equiv b \mod m} \leq \widetilde\Order(n).
\end{equation*}
We continue with some terminology: Call an integer $i$ is \emph{$\alpha$-heavy} if $\abs{\set{a \in A : a \equiv i \mod m}} > \alpha$, and \emph{$\alpha$-light} otherwise.

\begin{claim}
Let $\alpha > 1$. The number of $\alpha$-heavy elements $a \in A$ is at most $\widetilde\Order(n / \alpha)$.
\end{claim}
\begin{proof}
The proof is by a simple calculation. We express the number of $\alpha$-heavy elements in $A$ as
\begin{gather*}
    \sum_{\substack{i \in [m]\\\text{$i$ is heavy}}} \abs{\set{a \in A : a \equiv i \mod m}} \\
    \qquad\leq \sum_{\ell = \floor{\log \alpha}}^\infty \sum_{\substack{i \in [m]\\\text{$i$ is $2^\ell$-heavy}\\\text{$i$ is $2^{\ell+1}$-light}}} \abs{\set{a \in A : a \equiv i \mod m}} \\
    \qquad\leq \sum_{\ell = \floor{\log \alpha}}^\infty \sum_{\substack{i \in [m]\\\text{$i$ is $2^\ell$-heavy}}} 2^{\ell+1}
\intertext{Now observe that each bucket $i \in [m]$ which is $2^\ell$-heavy contributes at least \smash{$\binom{2^\ell}{2} = \Omega(2^{2\ell})$} many collisions. Since the total number of collisions is bounded by $\widetilde\Order(n)$, there can be at most $\widetilde\Order(n / 2^{2\ell})$ many such heavy buckets, and thus}
    \qquad\leq \sum_{\ell = \floor{\log \alpha}} \widetilde\Order(2^{\ell} \cdot \tfrac{n}{2^{2\ell}}) \\
    \qquad\leq \widetilde\Order(n / \alpha). \qedhere
\end{gather*}
\end{proof}

Coming back to the lemma, let us simply call $i$ \emph{heavy} if it is $n^\delta$-heavy for some parameter~\makebox{$\delta > 0$} to be determined later, and \emph{light} otherwise. We consider two cases, depending on whether the 3SUM solution contains a 3SUM solution with a heavy element, or whether all 3SUM solutions consist of light elements only. For the former case, since there are at most $\widetilde\Order(n^{1-\delta})$ many heavy elements $a \in A$, we can simply brute-force over all solutions involving a heavy element in time~\makebox{$\widetilde\Order(n^{1-\delta} \cdot n) = \widetilde\Order(n^{2-\delta})$}. Hence, for the rest of the proof we focus on the latter case and prove that we can detect 3SUM solutions involving only light elements.

Let $(x, y, z) \in [n^\delta]^3$. For each such triple, we construct a trichromatic Convolution 3SUM instance $(X_x, Y_y, Z_z)$ of length $2m$ as follows. (Reducing that trichromatic instance further to a monochromatic instance is standard.) For $i \in [2m]$, let $X_x[i] := \frac{a - i}{m}$ where $a$ is the $x$-th element in~$A$ with \smash{$a \equiv i \mod m$} (according to an arbitrary but fix order). Similarly, let \smash{$Y_y[j] := \frac{b - j}{m}$} where~$b$ is the $y$-th element in $A$ with $b \equiv j \mod m$ and let \smash{$Z_z[k] := \frac{c-k}{m}$} where $c$ is the $z$-th element in~$A$ with $z \equiv k \mod m$. If there happens to be no $x$-th (or $y$-th or $z$-th) element, simply put a dummy value that cannot be part of any solution. It is easy to verify that all entries are integers, and further have size at most $n^3 / m = \Order(n^2)$.

Moreover, there is a solution in any of the Convolution 3SUM instances $(X_x, Y_y, Z_z)_{x, y, z \in [n^\delta]}$ if and only if the given 3SUM instance $A$ has a light solution. For the one direction, suppose that there are light elements~\makebox{$a, b, c \in A$} with $a + b = c$. Let $i := x \bmod m$, $j := y \bmod m$ and $k := i + j$; clearly we have $k \equiv x \mod m$. Then there are $x, y, z \in [n^\delta]$ such that $a$ is the $x$-th element in $A$ with $a \equiv i \mod m$, $b$ is the $y$-th element in $A$ with $b \equiv j \mod m$, and $c$ is the $z$-th element in~$A$ with $c \equiv k \mod m$. By definition we have that $i + j = k$ and that
\begin{equation*}
    X_x[i] + Y_y[j] = \frac{a - i + b - j}{m} = \frac{c - k}{m} = Z_z[k].
\end{equation*}
The other direction is symmetric.

We have constructed $n^{3\delta}$ many Convolution 3SUM instances over a quadratic universe, and solving all of these takes time $\Order(n^{3\delta} \cdot n^{2-\epsilon})$ by our initial assumption. The total time is thus $\Order(n^{2-\delta} + n^{2-\epsilon+3\delta})$ which is subquadratic by choosing $\delta = \frac{\epsilon}{4} > 0$.
\end{proof}

\subsection{Universe Reduction for 3SUM Listing}
As a final simple corollary we include the following deterministic universe reduction for the 3SUM Listing problem (where the goal is to list all solutions):

\begin{lemma}[3SUM Listing over Small Universes]
For any $1 < \mu < 2$ and $\epsilon > 0$, there is no deterministic algorithm that given a size-$n$ 3SUM instance $A \subseteq [n^\mu]$ with $\Order(n^{3-\mu})$ solutions lists all solutions in time $\Order(n^{2-\epsilon})$, unless the deterministic 3SUM hypothesis fails.
\end{lemma}
\begin{proof}
This lemma is basically an immediate consequence of \cref{lem:hashing-few-pseudo-solutions}. The only minor complication is how to remove the overhead $\delta$.

Let $A \subseteq [n^c]$ be a given 3SUM instance. Using \cref{lem:hashing-few-pseudo-solutions} (with parameters $\mu$, $0 < \delta \leq \mu - 1$ to be determined, and $g = 1$) we can compute a modulus $m \in [n^{\mu}, 2n^{\mu}]$ such that the number of pseudo-solutions is bounded by
\begin{equation*}
    S := \#\set{(a, b, c) \in A^3 : a + b \equiv c \mod{m}} \leq \widetilde\Order(n^{3-\mu+\delta}).
\end{equation*}
We construct the 3SUM Listing instance $A' = \set{a \bmod m, (a \bmod m) + m : a \in A} \cup B$, where $B$ is a set of~$n^{1+\delta}$ dummy elements that cannot be part of a solution (e.g., \makebox{$B = \set{10m + i : i \in [n^{1+\delta}]}$}). Clearly the solutions in $A'$ stand in $\Order(1)$-to-$1$ correspondence to the pseudo-solutions of $A$. Therefore, given a list of all $\Order(S)$ solutions of $A'$ we can decide the original instance by testing whether one of the pseudo-solutions is proper. Note that $n' = |A'| = \Theta(n^{1+\delta})$, and thus $A' \subseteq [\Order(m)] \subseteq [(n')^{\mu}]$ and the number of solutions is $\Order(S) = \widetilde\Order(n^{3-\mu+\delta}) = \Order((n')^{3-\mu})$. Applying the efficient listing algorithm takes subquadratic time $\Order((n')^{2-\epsilon}) = \Order(n^{2+2\delta-\epsilon})$, by choosing $\delta := \min\set{\frac{\epsilon}{3}, \mu - 1} > 0$.
\end{proof}
\section{Deterministic Lower Bounds for Set Disjointness and Set Intersection} \label{sec:set-disjointness}
In this section we give our deterministic lower bounds for the (offline) Set Disjointness and (offline) Set Intersection problems (\cref{thm:set-disjointness,thm:set-intersection}). Let us recall their definitions first.

\SetDisjointnessProblem*

\SetIntersectionProblem*

The Set Disjointness has two ``trivial'' solutions: On the one hand, we can solve each query in time $\Order(s)$ thus taking time $\Order(q s)$ in total. On the other hand, we can encode all queries as a matrix multiplication problem of size $N \times s \times N$. If matrix multiplication is in linear time (i.e., if~$\omega = 2$), this leads to an algorithm in time $(N^2 + Ns)^{1+\order(1)}$. While the first approach works out of the box also for the Set Intersection problem, the latter algorithm does not immediately apply.

Kopelowitz, Pettie and Porat~\cite{KopelowitzPP16} have proved that these algorithms are best-possible---at least if 3SUM cannot be solved in randomized subquadratic time. In what follows, we derandomize their reductions.\footnote{Let us remark that in their paper~\cite{KopelowitzPP16} they even optimize their reductions with respect to logarithmic factors; here, since we are only interested in subquadratic-time 3SUM-hardness, we consider a coarser-grained view.}

We start with the following trivial reduction from Set Intersection (for small thresholds~$t$) to Set Disjointness. We will later rely on this reduction for $t = n^\epsilon$.

\begin{observation}[Trivial Reduction from Set Intersection to Set Disjointness] \label{thm:set-intersection-to-disjointness}
If the offline Set Disjointness problem is in deterministic time $f(U, N, s, q)$, then the offline Set Intersection problem with a threshold of $t$ per query is in time $\Order(f(U, Nt^2, \frac s{t^2}, qt^4) + \frac{qs}{t}) = \Order(t^6 \cdot f(U, N, s, q) + \frac{qs}{t})$.
\end{observation}
\begin{proof}
Given a Set Intersection instance $S_1, \dots, S_N$, let us arbitrarily partition each set $S_i$ into~$t^2$ subsets $S_{i, 1}, \dots, S_{i, t^2}$ of size $\Order(s / t^2)$. We replace each original by $t^4$ queries that respectively compare all possible combination of subsets. In this way we clearly obtain an equivalent Set Intersection instance, with $Nt^2$ many sets of size $\Order(s / t^2)$, over the same universe size $U$ and with~$qt^4$ many queries. We run the fast Set Disjointness algorithm on that instance. Our task is to report, for each query $(i, j) \in Q$ in the original Set Intersection instance, up to $t$ elements from the intersection~$S_i \cap S_j$. The outcome of the Set Disjointness call yields a subset $Q_{i, j} \subseteq [t^2]^2$ of positions~$(i', j')$ such that~$S_{i, i'} \cap S_{j, j'}$ is nonempty. Among these pairs in $Q_{i, j}$, select up to $t$ arbitrarily. For each selected pair~$(i', j')$ we evaluate~$S_{i, i'} \cap S_{j, j'}$ by scanning through all elements of~$S_{i, i'}$. Each such scan reveals at least one element of the intersection $S_i \cap S_j$, and therefore indeed after $t$ repetitions we have achieved our goal.

The algorithm takes time $f(U, Nt^2, \frac{s}{t^2}, qt^4)$ for the call to the Set Disjointness oracle. Afterwards, we explicitly compute the set intersections of $q \cdot t$ sets of size $\Order(s / t^2)$ each, taking time~\makebox{$\Order(q \cdot t \cdot s / t^2) = \Order(\frac{qs}{t})$} in total.
\end{proof}

\thmsetdisjointness*

Let us remark that the case $\alpha = 0$ is trivial, since the input size is already $\Omega(n^2)$.

\begin{proof}
Suppose there is a Set Disjointness algorithm with parameters as specified in the lemma statement. We will design a subquadratic algorithm for 3SUM. Given a 3SUM instance $A$, we proceed in the following steps:
\begin{enumerate}
    \item We apply the deterministic self-reduction for 3SUM. Specifically, we apply \cref{lem:dominance} to partition $A$ into $g = n^{\alpha}$ groups $A_1, \dots, A_g$ of size $\Order(n^{1-\alpha})$, such that only a subset $R \subseteq [g]^3$ of at most $|R| = \Order(n^{2\alpha})$ group triples is relevant.
    \item We apply the deterministic hashing lemma to find a modulus $m$ under which we have only few false positives in the relevant groups. Specifically, we apply \cref{lem:hashing-few-pseudo-solutions} for the 3SUM instances $(A_i \cup A_j \cup A_k)_{(i, j, k) \in R}$ (of size $\Order(n^{1-\alpha})$ each) and for $\mu = 2 - 2\delta$ and some parameter~$\delta > 0$ to be specified later. This either reports that one of the instances contains a 3SUM solution (in which case we can stop), or the lemma returns a modulus $m = \Order((n^{1-\alpha})^{\mu}) = \Order(n^{2-2\alpha})$ satisfying that
    \begin{multline*}
        \sum_{(i, j, k) \in R} \#\set{(a, b, c) \in A_i \times A_j \times A_k : a + b \equiv c \mod m} \\\leq |R| \cdot \Order((n^{1-\alpha})^{1+3\delta}) \cdot (\log n)^{\Order(1/\delta)} = \widetilde\Order(n^{1+\alpha+3\delta}).
    \end{multline*}
    Let us for simplicity assume that $\sqrt m$ is an integer; otherwise replace all the following occurrences of $\sqrt m$ with an appropriate rounding of $\sqrt m$.
    \item We are ready to construct the Set Disjointness instance: For each $i \in [g]$ and for all~\makebox{$x, y \in [\sqrt m]$}, we construct the following two sets:
    \begin{align*}
        B_{i, x} &= \set{(b + x \sqrt m) \bmod m : b \in A_i}, \\
        C_{i, y} &= \set{(c + y) \bmod m : c \in A_i}.
    \end{align*}
    The Set Disjointness instance consists of the sets $\set{B_{i, x}}_{i, x} \cup \set{C_{i, y}}_{i, y}$. The queries are defined as follows: For each $(i, j, k) \in R$ and for each $a \in A_i$, we add a query. Namely, let $x, y \in [\sqrt m]$ be such that $a \equiv x \sqrt m - y \mod m$; then we query whether the intersection $B_{j, x} \cap C_{k, y}$ is empty.
\end{enumerate}
Let us take a break here from the algorithm, and for now give an explanation of the Set Disjointness instance. We claim that the elements in the queried set intersections $B_{j, x} \cap C_{k, y}$ are in one-to-one correspondence to the set of \emph{pseudo-solutions}
\begin{equation*}
    \bigcup_{(i, j, k) \in R} \set{(a, b, c) \in A_i \times A_j \times A_k : a + b \equiv c \mod m}.
\end{equation*}
Indeed, for any pseudo-solution $(a, b, c) \in A_i \times A_j \times A_k$, let $x, y \in [\sqrt m]$ be such that $a \equiv x \sqrt m - y \mod m$. Then note that $(b + x \sqrt m) \bmod m \in B_{j, x}$ and that $(c + y) \bmod m \in C_{k, y}$. Moreover, since $x \sqrt m - y + b \equiv a + b \equiv c \mod m$, these two elements are in fact identical. The converse direction is analogous.

With this in mind, our goal is to list \emph{all} elements in the intersections $B_{j, x} \cap C_{k, y}$ of all queries. In this way we obtain the set of all pseudo-solutions and we can easily check whether a 3SUM solution is among them. However, the Set Disjointness algorithm is not required to report the elements in the intersections. Instead, we proceed as follows:
\begin{enumerate}
    \setcounter{enumi}{3}
    \item Using the reduction from \cref{thm:set-intersection-to-disjointness} with threshold $t = n^\rho$ (for some parameter $\rho$ to be determined later), we list up to $t$ elements from each queried intersection.
    \item For each query $B_{j, x} \cap C_{k, y}$ with $|B_{j, x} \cap C_{k, y}| \geq t$ (i.e., for which we have exhausted the threshold in the previous listing step), we explicitly enumerate all elements in the intersection by scanning once through $B_{j, x}$. 
\end{enumerate}
After these two steps, we have clearly listed all elements in the queried intersections. As argued before, this completes the 3SUM algorithm.

It remains to argue that the running time is subquadratic as claimed. Let us analyze the five steps of the algorithm individually:
\begin{enumerate}
    \item Takes time $\widetilde\Order(n + g^2) = \widetilde\Order(n + n^{2\alpha})$ by \cref{lem:dominance}.
    \item Takes time $\widetilde\Order(g (n^{1-\alpha})^{\max(1, \mu)+\delta}) = \Order(n^{2\alpha} (n^{1-\alpha+\delta} + n^{(1-\alpha)(2-\delta)})) = \Order(n^{1+\alpha+\delta} + n^{2-\delta(1-\alpha)})$ by \cref{lem:hashing-few-pseudo-solutions}.
    \item The whole construction takes linear time in the size of the Set Disjointness instance. This is negligible in comparison to step 4 (where we actually solve the instance).
    \item To analyze the running time of this step, we first analyze the parameters of the constructed Set Disjointness instance: Observe that
    \begin{align*}
        U &= m = \Order(n^{2-2\alpha}), \\
        N &= 2g \sqrt{m} = \Order(n^{\alpha + 1 - \alpha}) = \Order(n), \\
        s &= \Order(n^{1-\alpha}), \\
        q &= |R| \cdot \Order(n^{1-\alpha}) = \Order(n^{1+\alpha}).
    \end{align*}
    Since these are exactly as claimed in the lemma statement, the efficient Set Disjointness algorithm can solve this instance in time $\Order(n^{2-\epsilon})$ for some $\epsilon > 0$. Due to the overhead of \cref{thm:set-intersection-to-disjointness}, this step runs in total time $\Order(t^6 \cdot n^{2-\epsilon} + \frac{qs}{t}) = \Order(n^{2-\epsilon+6\rho} + n^{2-\rho})$.
    \item To analyze the running time of this step, first recall that there are at most $\widetilde\Order(n^{1+\alpha+3\delta})$ many pseudo-solutions (modulo $m$). Moreover, the pseudo-solutions are in one-to-one correspondence to the queried set intersections. Since we only query the intersections containing at least~\makebox{$t = n^\rho$} elements, there can be at most $\widetilde\Order(n^{1+\alpha+3\delta-\rho})$ that have not been completely answered in the previous step. Each such remaining query takes time $\Order(s) = \Order(n^{1-\alpha})$, hence this step takes time $\widetilde\Order(n^{2 + 3\delta - \rho})$ in total.
\end{enumerate}
Summing over all contributions, the total running time is
\begin{equation*}
    \widetilde\Order(n + n^{2\alpha} + n^{1+\alpha + \delta} + n^{2-\delta(1-\alpha)} + n^{2-\epsilon+6\rho} + n^{2-\rho} + n^{2+3\delta - \rho}).
\end{equation*}
Recall that $0 \leq \alpha < 1$. Thus by picking e.g.\ $\rho = \frac{\epsilon}{7} > 0$ and $\delta = \min\set{\frac{1-\alpha}{2}, \frac{\rho}{4}} > 0$ this running time becomes subquadratic, \smash{$\Order(n^{2-\epsilon'})$} for some $\epsilon' > 0$.
\end{proof}

Next, let us turn to the analogous proof for Set Intersection. The idea is similar, but requires some tweaks (as in the original paper~\cite{KopelowitzPP16}).

\thmsetintersection*

With this theorem we prove that Set Intersection has a deterministic matching lower bound whenever the queries are expected to produce at least $\Omega(1)$ elements on average, possibly much more. We remark however that there are some non-trivial cases which are covered by the randomized reduction~\cite{KopelowitzPP16} which we cannot cover here---namely if $1 - \alpha < \beta \leq 1$. This parameterization describes a setting where we expect the total size of the intersections to be much smaller than the total number of queries.

\begin{proof}
This proof is overall very similar to the proof of \cref{thm:set-disjointness}. We will thus be more concise here. Assume that Set Disjointness with parameters as specified in the lemma statement is in subquadratic time; we design a subquadratic algorithm for a given 3SUM instance $A$:
\begin{enumerate}
    \item We apply the deterministic self-reduction for 3SUM. Specifically, we apply \cref{lem:dominance} to partition $A$ into $g = n^\alpha$ groups $A_1, \dots, A_g$ of size $\Order(n^{1-\alpha})$, such that only a subset $R \subseteq [g]^3$ of at most $|R| = \Order(n^{2\alpha})$ group triples is relevant.
    \item We apply the deterministic hashing lemma to find a modulus $m$ under which we have only few false positives in the relevant groups. Specifically, we apply \cref{lem:hashing-few-pseudo-solutions} for the 3SUM instances $(A_i \cup A_j \cup A_k)_{(i, j, k) \in R}$ (of size $\Order(n^{1-\alpha})$ each) and for $\mu = \frac{1-\alpha+\beta}{1-\alpha} - 2\delta$ and some parameter $\delta > 0$ to be specified later. This either reports that one of the instances contains a 3SUM solution (in which case we can stop), or the lemma returns a modulus $m = \Order((n^{1-\alpha})^\mu) = \Order(n^{1-\alpha+\beta})$ satisfying that
    \begin{multline*}
        \sum_{(i, j, k) \in R} \#\set{(a, b, c) \in A_i \times A_j \times A_k : a + b \equiv c \mod m} \\\leq |R| \cdot \Order((n^{1-\alpha})^{3-\mu+\delta}) \cdot (\log n)^{\Order(1/\delta)} = \widetilde\Order(n^{2\alpha + 3 - 3\alpha - 1 + \alpha - \beta + 3\delta}) = \widetilde\Order(n^{2-\beta+3\delta}). 
    \end{multline*}
    Let us for simplicity assume that $\sqrt m$ is an integer; otherwise replace all the following occurrences of $\sqrt m$ with an appropriate rounding of $\sqrt m$.
    \item We construct the Set Intersection instance exactly as in \cref{thm:set-disjointness}: For each $i \in [g]$ and for all $x, y \in [\sqrt m]$, we construct the following two sets:
    \begin{align*}
        B_{i, x} &= \set{(b + x \sqrt m) \bmod m : b \in A_i}, \\
        C_{i, y} &= \set{(c + y) \bmod m : c \in A_i}.
    \end{align*}
    The Set Intersection instance consists of the sets $\set{B_{i, x}}_{i, x} \cup \set{C_{i, y}}_{i, y}$. The queries are defined as follows: For each $(i, j, k) \in R$ and for each $a \in A_i$, we query the intersection $B_{j, x} \cap C_{k, y}$, where $x, y \in [\sqrt m]$ are such that $a \equiv x \sqrt m - y \bmod m$. As before, there is a natural one-to-one correspondence between the queried elements in the set intersections and the pseudo-solutions modulo $m$.
    \item This is where the proof diverges. We can simply solve the Set Intersection instance using the oracle. For each reported element in a one of the set intersections, we test in $\Order(1)$ time whether the corresponding pseudo-solution is a proper 3SUM solution.
\end{enumerate}
This completes the description of the algorithm. The correctness is immediate from the given in-text descriptions. It remains to bound the running time, for which we again analyze the four steps individually:
\begin{enumerate}
    \item Takes time $\widetilde\Order(n + g^2) = \widetilde\Order(n + n^{2\alpha})$ by \cref{lem:dominance}.
    \item Takes time
    \begin{equation*}\widetilde\Order(g (n^{1-\alpha})^{\max(1, \mu)+\delta}) = \Order(n^{2\alpha} (n^{1-\alpha+\delta} + n^{(1-\alpha)(\frac{1-\alpha+\beta}{1-\alpha}-\delta)})) = \Order(n^{1+\alpha+\delta} + n^{1+\alpha+\beta-\delta(1-\alpha)})
    \end{equation*}
    by \cref{lem:hashing-few-pseudo-solutions}.
    \item The whole construction takes linear time in the size of the Set Intersection instance. This is negligible in comparison to step 4 (where we actually solve the instance).
    \item To analyze the running time of this step, we first analyze the parameters of the constructed Set Intersection instance: Observe that
    \begin{align*}
        U &= m = \Order(n^{1-\alpha+\beta}), \\
        N &= 2g \sqrt{m} = \Order(n^{\alpha + \frac12 - \frac\alpha2 + \frac\beta2}) = \Order(n^{\frac12 + \frac\alpha2 + \frac\beta2}), \\
        s &= \Order(n^{1-\alpha}), \\
        q &= |R| \cdot \Order(n^{1-\alpha}) = \Order(n^{1+\alpha}).
    \end{align*}
    These are exactly as claimed in the lemma statement. Recall further that the total number of elements in all set intersections (i.e., the output size) is exactly the number of pseudo-solutions modulo $m$. By the second step, this number is bounded by $\widetilde\Order(n^{2-\beta+3\delta})$. This is slightly more than we have claimed in the lemma statement, however, we can simply substitute~$n$ for~\makebox{$n' = n^{1+3\delta}$}; then the constructed Set Intersection instance has the claimed parameters with respect to $n'$. The oracle thus takes time $\Order((n')^{2-\epsilon}) = \Order(n^{(1+3\delta)(2-\epsilon)}) = \Order(n^{2-\epsilon+6\delta})$ to solve the instance.
\end{enumerate}
Summing over all contributions, the total running time is
\begin{equation*}
    \widetilde\Order(n + n^{2\alpha} + n^{1+\alpha+\delta} + n^{1+\alpha+\beta-\delta(1-\alpha)} + n^{2-\epsilon+6\delta}).
\end{equation*}
Recall that $0 \leq \alpha < 1$ and $0 \leq \beta \leq 1 - \alpha$. Thus, by picking e.g.\ $\delta = \min\set{\frac{1-\alpha}{2}, \frac{\epsilon}{7}}$ the running time becomes truly subquadratic, \smash{$\Order(n^{2-\epsilon'})$} for some $\epsilon' > 0$.
\end{proof}




    

    




\bibliographystyle{alphaurl}
\bibliography{references}

\appendix
\section{Deterministic Lower Bound for Mono Convolution} \label{sec:mono-conv}
In this section we prove that the Mono Convolution problem (defined as follows) is deterministically 3SUM-hard.

\begin{problem}[Mono Convolution]
Given three sequences $X, Y, Z \in \Int^n$, determine for each $k \in [n]$ whether there exists a pair $i, j \in [n]$ with $i + j = k$ and $X[i] = Y[j] = Z[k]$.
\end{problem}

It is known that this problem can be solved in time $\widetilde\Order(n^{3/2})$, and Lincoln, Polak and Vassilevska Williams~\cite{LincolnPW20} proved that any polynomial improvement over this baseline algorithm is \emph{equivalent} to a subquadratic algorithm for 3SUM (by means of a randomized reduction). Our goal here is to deranomize their reduction:

\begin{lemma}[Deterministic Mono Convolution Hardness] \label{lem:mono-conv}
Unless the deterministic 3SUM hypothesis fails, there is no deterministic algorithm for the Mono Convolution problem running in time $\Order(n^{2-\epsilon})$, for any $\epsilon > 0$.
\end{lemma}

The randomized reduction it is based on the insight that an \emph{$s$-sparse} Mono Convolution instance is equivalent to a 3SUM instance of size $s$ and universe size $U$ (by viewing the sequences $X, Y, Z$ as the indicator vectors of the three sets). The crucial part in the reduction is that we can combine several 3SUM instances into one \emph{dense} Mono Convolution instance. To this end, one can choose \emph{random offsets} such that in every coordinate we have only few collisions. Besides the usual 3SUM machinery, we therefore have to rely on the following deterministic load balancing lemma:

\begin{lemma}[Deterministic Load Balancing] \label{lem:load-balancing}
Let $\delta > 0$, and let $A_1, \dots, A_N \subseteq [U]$ be sets of size at most $S$. We can compute a partition of each set into $A_i = A_i' \cup A_i''$ and shifts $s_1, \dots, s_N \in [U]$ such that
\begin{enumerate}
    \item For all $a \in [3U]$, we have $\abs{\set{i : a \in A_i' + s_i}}, \abs{\set{i : a \in A_i' + 2s_i}} \leq N^{5\delta}$.
    \item $\sum_i |A_i''| \leq \Order(N^{2-\delta} S^2 / U)$. 
\end{enumerate}
The algorithm runs in deterministic time $\widetilde\Order(N S + N^{1-\delta} U)$. 
\end{lemma}

We postpone the proof of \cref{lem:load-balancing} to \cref{sec:mono-conv:sec:load-balancing}, and instead prove \cref{lem:mono-conv} here.

\begin{proof}[Proof of \cref{lem:mono-conv}]
Assume that there is an algorithm in time $\Order(n^{3/2-\epsilon})$ for the Mono Convolution problem. We design a subquadratic algorithm to decide a given 3SUM instance $A$ of size $n$. Consider the following five steps:
\begin{enumerate}
    \item We view $A$ as an instance of All-Numbers 3SUM and use the universe reduction (\cref{lem:univ-reduction-an-3sum}). We can therefore assume that $A \subseteq [n^2]$.
    \item We apply the deterministic self-reduction for All-Numbers 3SUM (\cref{lem:self-reduction-an-3sum} with parameter $g = n^{1/3}$) and reduce $A$ to $N = \Order(n^{2/3})$ All-Numbers 3SUM instances $A_1, \dots, A_N$ of size $\Order(n^{2/3})$ each, and over the univese $[\Order(n^{4/3})]$.
    \item We apply the load balancing lemma (with parameters~\makebox{$N = S = \Theta(n^{2/3}), U = \Theta(n^{4/3})$} and some parameter $\delta > 0$ to be determined later). In this way we obtain shifts $s_1, \dots, s_N$ and partitions~\makebox{$A_i = A_i' \cup A_i''$}. Let~\makebox{$L := N^{5\delta}$}. The lemma guarantees that for each $x \in [3U]$ we have $\abs{\set{i : a \in A_i' + s_i}}, \abs{\set{i : a \in A_i' + 2s_i}} \leq L$.
    \item Let $(x, y, z) \in [L]^3$. We construct a Mono Convolution instance $(X_x, Y_y, Z_z)$ as follows: We assign $X_x[a] := i$ whenever $i$ is the $x$-th index satisfying  that $a \in A_i' + s_i$ (according to an arbitrary but fix order). We similarly assign $Y_y[a] := i$ whenever $i$ is the $y$-th index satisfying~\makebox{$a \in A_i' + s_i$} and $Z_z[x] := i$ whenever $i$ is the $z$-th index satisfying $a \in A_i' + 2s_i$. For each $(x, y, z) \in [L]^3$ we solve $(X_x, Y_y, Z_z)$ using the efficient Mono Convolution algorithm. Whenever it reports $c$ (claiming that there are $a, b$ with $X_x[a] = Y_y[b] = Z_z[c] =: i$), we report that $c - 2s_i$ is a 3SUM solution in $A_i$.
    \item Finally, for each $i \in [N]$ test by brute-force all 3SUM solutions in $A_i$ with at least one element in $A_i''$.
\end{enumerate}

\subparagraph{Correctness.}
We start proving that this algorithm is correct. Clearly the universe reduction and self-reduction steps are correct, and it suffices to prove that we correctly solve the All-Numbers 3SUM instances~\makebox{$A_1, \dots, A_N$} in steps 4 and 5. In step 5 we are guaranteed to find all 3SUM solutions involving elements from $A_i''$, thus in step 4 it is sufficient to only consider the sets $A_1', \dots, A_N'$.

We prove this statement in two directions: First, suppose that there are $a', b', c' \in A_i'$ such that~\makebox{$a' + b' = c'$}. Let $a = a' + s_i$, $b = b' + s_i$ and $c = c' + 2s_i$. The load balancing lemma gives that $\abs{\set{i : a \in A_i' + s_i}}, \abs{\set{i : a \in A_i' + 2s_i}} \leq L$, and thus there are $x, y, z \in [L]$ such that $X_x[a] = i$, $Y_y[b] = i$ and $Z_z[c] = i$. It follows that $c$ is reported in the Mono Convolution instance $(X_x, Y_y, Z_z)$ and we correctly report $c' = c - 2s_i$. 

For the other direction, suppose that we report $c - 2s_i$. This only happens if there are $a, b$ with $a + b = c$ and $X_x[a] = Y_y[b] = Z_z[c]$. By construction, this means that there are elements~\makebox{$a' = a - s_i, b' = b - s_i, c' = c - 2s_i \in A_i'$}. Note that $a', b', c'$ is indeed a 3SUM solution.

\subparagraph{Running Time.}
It suffices to analyze the running time. Let us analyze the steps individually:
\begin{enumerate}
    \item The reduction to a quadratic-size universe runs in subquadratic time (\cref{lem:univ-reduction-an-3sum}).
    \item The self-reduction takes time $\widetilde\Order(n + g^2) = \widetilde\Order(n^{4/3})$ (\cref{lem:self-reduction-an-3sum}).
    \item The load balancing lemma takes time $\widetilde\Order(N S + N^{1-\delta} U) = \widetilde\Order(n^{2-2\delta/3})$ (\cref{lem:load-balancing}).
    \item We consider $L^3 = \Order(n^{15\delta})$ many triples $(x, y, z)$. For each such triple, constructing the Mono Convolution instance is negligible. Solving each instance with the efficient Mono Convolution algorithm takes time $\Order(U^{3/2-\epsilon}) = \Order(n^{2 - \epsilon})$, so the total time is $\Order(n^{2-\epsilon+15\delta})$.
    \item In the brute-forcing step it suffices to enumerate all pairs in $\bigcup_i (A_i'' \times A_i)$. The running time is thus bounded by $\Order(S \cdot \sum_i |A_i''|) = \Order(N^{2-\delta} S^3 / U) = \Order(n^{2-2\delta/3})$ (\cref{lem:load-balancing}).
\end{enumerate}
Summing over all contributions, the total running time is bounded by
\begin{equation*}
    \widetilde\Order(n^{4/3} + n^{2-2\delta/3} + n^{2-\epsilon+15\delta}).
\end{equation*}
By picking $\delta = \frac{\epsilon}{16} > 0$, this running time becomes truly subquadratic.
\end{proof}

\subsection{Deterministic Load Balancing} \label{sec:mono-conv:sec:load-balancing}
It remains to prove the deterministic load balancing lemma, \cref{lem:load-balancing}. We will do so in two steps: First, we prove a weaker version of the lemma claiming an algorithm running in time $\widetilde\Order(N U)$ (see \cref{lem:load-balancing-slow}). We then derive \cref{lem:load-balancing} from that weaker lemma.

\begin{lemma}[Slow Load Balancing] \label{lem:load-balancing-slow}
Let $\delta > 0$, and let $A_1, \dots, A_N \subseteq [U]$ be sets of size at most~$S$. We can compute a partition of each set into $A_i = A_i' \cup A_i''$ and shifts $s_1, \dots, s_N \in [U]$ such that
\begin{enumerate}
    \item For all $a \in [3U]$, we have $\abs{\set{i : a \in A_i' + s_i}}, \abs{\set{i : a \in A_i' + 2s_i}} \leq N^{\delta}$.
    \item $\sum_i |A_i''| \leq \Order(N^{2-\delta} S^2 / U)$.
\end{enumerate}
The algorithm runs in deterministic time $\widetilde\Order(N U)$. 
\end{lemma}

\begin{proof}
Recall that a feasible randomized solution is to pick the shifts $s_1, \dots, s_N$ uniformly at random. We follow the method of conditional expectations: We select the shifts $s_1, \dots, s_N \in [U]$ step by step while making sure that each choice is \emph{at least as good as what we expect from a random choice.} Specifically, consider the following algorithm: We execute stages $i \gets 1, \dots, N$, and at the $i$-th stage (where we have already fixed $s_1, \dots, s_{i-1}$), we select~\makebox{$s_i \gets s$} to be the minimizer of
\begin{equation*}
    M_i := \min_{s \in [U]} \sum_{j < i} \parens*{\vphantom{\Big(}|(A_j + s_j) \cap (A_i + s)| + |(A_j + 2s_j) \cap (A_i \cap 2s)|}. 
\end{equation*}

This algorithm determines the offsets $s_1, \dots, s_N$, and it remains to compute the partitions. We say that $a \in [3U]$ is \emph{$\alpha$-heavy} if $\abs{\set{i : a \in A_i + s_i}}, \abs{\set{i : a \in A_i + 2s_i}} > \alpha$, and \emph{$\alpha$-light} otherwise. We also say that $a$ is \emph{heavy} if it is $N^\delta$-heavy and \emph{light} otherwise. Using this terminology, we assign
\begin{align*}
    A_i' &= \set{b \in A_i : \text{$b + s_i$ is light}}, \\
    A_i'' &= \set{b \in A_i : \text{$b + s_i$ is heavy}}.
\end{align*}
This completes the description of the algorithm. In the following we prove that it is correct and efficient.

\subparagraph{Proof of Property 1.}
The first property is trivial by the way we assigned the sets $A_i', B_i', C_i'$: By definition we have $\abs{\set{i : a \in A_i' + s_i}}, \abs{\set{i : a \in A_i' + 2s_i}} \leq N^\delta$ for all $a \in [3U]$.

\subparagraph{Proof of Property 2..}
The second property requires more work. We prove the statement in two steps, starting with the following claim:

\begin{claim}
$M_i \leq \Order(i S^2 / U)$.
\end{claim}
\begin{proof}
In order to obtain an upper bound on $M_i$, pretend that we choose $s$ uniformly at random. We then bound $M_i$ by the \emph{expected} number of collisions:
\begin{gather*}
    M_j \leq \Ex_{s} \brackets*{\sum_{j < i} \parens*{|(A_j + s_j) \cap (A_i + s)| + |(A_j + 2s_j) \cap (A_i \cap 2s)|}} \\
    \qquad\leq \sum_{k < j} \,\Ex_{s}\, \brackets*{\vphantom{\Big|}|(A_j + s_j) \cap (A_i + s)| + |(A_j + 2s_j) \cap (A_i \cap 2s)|}
\intertext{To bound the expected size of $(A_j + s_j) \cap (A_i + s)$, observe that any pair of elements collides only with probability at most $1/U$. Thus, since each of the sets has size at most $S$, the expected intersection size is at most $S^2 / U$:}
    \qquad\leq \sum_{j < i} \Order(S^2 / U) \\
    \qquad= \Order(i S^2 / U). \qedhere
\end{gather*}
\end{proof}

\begin{claim}
Let
\begin{equation*}
    C_a := \binom{\abs{\set{i : a \in A_i + s_i}}}{2} + \binom{\abs{\set{i : a \in A_i + 2s_i}}}{2}
\end{equation*}
denote the number of \emph{collisions} at $a$. Then $\sum_{a \in [3U]} C_a \leq \Order(N^2 S^2 / U)$.
\end{claim}
\begin{proof}
The proof is by a simple calculation:
\begin{gather*}
    \sum_{a \in [3U]} C_a \\
    \qquad= \sum_{a \in [3U]} \bigg(\sum_{\substack{j < i\\a \in A_i + s_i\\a \in A_j + s_j}} 1 + \sum_{\substack{j < i\\a \in A_i + 2s_i\\a \in A_j + 2s_j}} 1\bigg) \\
    \qquad\leq \sum_{j < i} \parens*{\vphantom{\Big(} |(A_j + s_j) \cap (A_i + s_i)| + |(A_j + 2s_j) \cap (A_i \cap 2s_i)|} \\
    \qquad\leq \sum_{i \in [N]} M_i \\
    \qquad\leq \sum_{i \in [N]} \Order(i S^2 / U) \\
    \qquad\leq \Order(N^2 S^2 / U). \qedhere
\end{gather*}
\end{proof}

\noindent
With these two claims in mind, we are ready to prove that $\sum_i |A_i''| \leq \Order(N^{2-\delta} S^2 / U)$. Consider the following calculation:
\begin{gather*}
    \sum_{i \in [N]} |A_i''| \\
    \qquad\leq \sum_{\substack{a \in [3U]\\\text{$a$ is heavy}}} \abs{\set{i : a \in A_i + s_i}} + \abs{\set{i : a \in A_i + 2s_i}} \\
    \qquad\leq \sum_{\ell = {\floor{\delta \log N}}}^{\log N} \sum_{\substack{a \in [3U]\\\text{$a$ is $2^\ell$-heavy}\\\text{$a$ is $2^{\ell+1}$-light}}} \abs{\set{i : a \in A_i + s_i}} + \abs{\set{i : a \in A_i + 2s_i}}\\
    \qquad\leq \sum_{\ell = {\floor{\delta \log N}}}^{\log N} \sum_{\substack{a \in [3U]\\\text{$a$ is $2^\ell$-heavy}}} \Order(2^\ell)
\intertext{Notice that each $2^\ell$-heavy element leads to at least $\binom{2^\ell}{2} = \Theta(2^{2\ell})$ many collisions. Since the total number of collisions is bounded by $\sum_a C_a \leq \Order(N^2 S^2 / U)$ (by the previous claim), there can be at most $\Order(N^2 S^2 / (2^{2\ell} U))$ many $2^\ell$-heavy elements:}
    \qquad\leq \sum_{\ell = {\floor{\delta \log N}}}^{\log N} \Order\parens*{2^\ell \cdot \frac{N^2 S^2}{2^{2\ell} U}} \\
    \qquad\leq \Order(N^{2-\delta} S^2 / U).
\end{gather*}

\subparagraph{Running Time.}
It suffices to prove that the algorithm can be implemented in time $\widetilde\Order(N U)$. Recall that we run $N$ stages, so it suffices to prove that each stage can be implemented in time~$\widetilde\Order(U)$. The relevant task is to find the minimizer $s$ of $M_i$. We solve this problem via the Fast Fourier Transform: Let $V$ be the vectors defined by $V[a] = \abs{\set{j < i : a \in A_j + s_j}}$ and let $W$ be the indicator vector of $-A_i$ (allowing negative indices for simplicity). Then our goal is to compute, for every $s \in [U]$,
\begin{equation*}
    \sum_{a \in [U]} V[a] \cdot W[s - a] = \sum_{\substack{a \in [U]\\s - a \in -A_i}} \abs{\set{j < i : a \in A_j + s_j}} = \sum_{j < i} |(A_j + s_j) \cap (A_i + s)|.
\end{equation*}
This is exactly a convolution computation that can be solved by the Fast Fourier Transform in near-linear time $\widetilde\Order(U)$.
\end{proof}

\begin{proof}[Proof of \cref{lem:load-balancing}]
Our strategy is to group the sets $A_1, \dots, A_n$ into \emph{buckets} $\bar A_1, \dots, \bar A_{\bar N}$ and to apply the previous lemma on these buckets. Specifically, let $\bar N := \ceil{N^{1-\delta}}$ and arbitrarily partition~$[N]$ into sets $I_1, \dots, I_{\bar N}$ of size at most $\ceil{N^{\delta}}$. We then assign
\begin{align*}
    \bar A_j = \bigcup_{i \in I_j} A_i.
\end{align*}
for $j \in [\bar N]$. Note that each of these buckets has size at most $\bar S := S \ceil{N^\delta}$. We apply \cref{lem:load-balancing-slow} with parameter $4\delta$ to obtain shifts $\bar s_1, \dots, \bar s_{\bar N}$ and partitions $\bar A_j = \bar A_j' \cup \bar A_j''$. For any $i \in I_j$, we then assign $s_i := s_j$, \smash{$A_i' := A_i \cap \bar A_j'$}, and \smash{$A_i'' := A_i \cap \bar A_j''$}.

We prove that this construction satisfies the claimed properties:
\begin{enumerate}
    \item For any $a \in [3U]$, we have that $\abs{\set{j : a \in \bar A_j' + \bar s_j}} \leq \bar N^{4\delta} \leq N^{4\delta}$. Since each bucket consists of at most $N^\delta$ original sets, it follows that \smash{$\abs{\set{i : a \in A_i' + s_i}} \leq N^{5\delta}$}.
    \item Similarly, since $\sum_j |\bar A_j''| \leq \Order(\bar N^{2-4\delta} \bar S^2 / U) = \Order(N^{2-2\delta} S^2 / U)$, and since each bucket consists of at most $N^\delta$ original sets, we have that \smash{$\sum_i |A_i''| \leq \Order(N^{2-\delta} S^2 / U)$}. 
\end{enumerate}

It remains to bound the running time. Constructing the buckets takes time $\Order(NS)$. Invoking the previous lemma takes time $\Order(\bar N U) = \Order(N^{1-\delta} U)$. Therefore, the total time is $\widetilde\Order(N S + N^{1-\delta} U)$ as claimed.
\end{proof}
\section{Deterministic Lower Bound for Local Alignment} \label{sec:local-alignment}

In this section we derandomize the known 3SUM-based hardness reduction for the local alignment problem by Abboud, Vassilevska Williams and Weimann~\cite{AbboudWW14}.

\begin{problem}[Local Alignment]
Given two strings $X, Y \in \Sigma^n$ over an alphabet $\Sigma$, and a scoring function $s : \Sigma \times \Sigma \to \Int$, compute the substrings $X', Y'$ of $X, Y$ respectively that have the same length and maximize $\sum_{i=1}^{|X'|} s(X'[i], Y'[i])$.
\end{problem}

\begin{theorem}[Deterministic Local Alignment Hardness] \label{thm:local-alignment}
Let $0 < \delta < 1$. Unless the deterministic 3SUM hypothesis fails, the Local Alignment problem on two strings of length $n$ over an alphabet of size~$\widetilde\Order(n^{1-\delta})$ cannot be solved in deterministic time $\Order(n^{2-\delta-\epsilon})$, for any $\epsilon > 0$.
\end{theorem}

The deterministic reduction closely follows the original reduction~\cite{AbboudWW14} except that we replace randomized hashing with a deterministic construction. Here, in fact, it is not even necessary to control the number of pseudo-solutions, and we use the few-collisions deterministic hashing (\cref{lem:hashing-few-collisions},~\cite{ChanL15}). The proof additionally relies on the following lemma due to Abboud, Lewi and Williams~\cite{AbboudLW14}. Throughout we write $[i, j] = \set{i, \dots, j}$.

\begin{lemma}[Immediate by {\cite[Lemma~3.2]{AbboudLW14}}] \label{lem:carry-guessing}
Let $U \geq 1$. Then there are integers~\makebox{$N = U^{\order(1)}$} and $d = \Order(\log U / \log\log U)$, and mappings $f_1, \dots, f_N : [-U, U] \to [-\ceil{\log U}, \ceil{\log U}]^d$ (each deterministically computable in time $U^{\order(1)}$) satisfying the following property: For all $a, b, c \in [-U, U]$,
\begin{equation*}
    a + b + c = 0 \qquad\text{if and only if}\qquad \exists i \in [N]: f_i(a) + f_i(b) + f_i(c) = 0.
\end{equation*}
\end{lemma}

\begin{proof}[Proof of \cref{thm:local-alignment}]
To obtain a contradiction under the deterministic 3SUM hypothesis, it suffices to design a $\Order(n^{2-\epsilon})$-time 3SUM algorithm over a universe of size $n^3$. For this proof it will be convenient to work with the following version of 3SUM: Given sets $A, B, C \subseteq [-n^3, n^3]$, test whether there exist $a \in A, b \in B, c \in C$ with $a + b + c = 0$. This alternative version is easily seen to be equivalent to our usual definition (by means of deterministic reductions).

We start with an application of \cref{lem:hashing-few-collisions} (with parameters $\mu = 1 - \delta$ and $\delta = \frac12$), which yields a modulus $m = \Theta(n^{1 - \delta})$ satisfying that the number of collisions is bounded as follows:
\begin{equation*}
    \#\set{a, b \in A \cup B \cup C : a \equiv b \mod m} \leq \widetilde\Order(n^{1+\delta}).
\end{equation*}
Let $\gamma$ be a parameter to be determined later. Note that there can be at most $\widetilde\Order(n^{1-\delta - \gamma})$ many buckets $i \in [m]$ with $\#\set{a \in A \cup B \cup C : a \equiv i \mod m} \geq t := n^{\delta + \gamma}$; we call these buckets \emph{heavy.} We can afford to brute-force all solutions involving an element from a heavy bucket in time~$\Order(n^{2-\gamma})$ in the same way as we did in many previous proofs; for the sake of brevity we omit this part and focus instead on detecting a solution contain only in \emph{light} buckets. We thus assume in the following that every bucket contains at most $t$ elements.

Next, we apply \cref{lem:carry-guessing} for $U = n^3$, $N = U^{\order(1)}$, $q = \ceil{\log U}$ and $d = \Order(\log n / \log\log n)$, to obtain functions $f_1, \dots, f_N : [-U, U] \to [-q, q]^d$. For every pair of numbers~\makebox{$(i, j) \in [N] \times [t]$} we create an instance of the Local Alignment problem (i.e.\ two strings $X, Y$ and a scoring function $s(\cdot, \cdot)$). The optimal solution to that instance will determine whether there are three numbers~\makebox{$a \in A, b \in B, c \in C$} such that 
\begin{enumerate}[label=(\roman*)]
    \item $c$ is the $j$-th element in its bucket,
    \item $a + b + c \equiv 0 \mod m$,
    \item $f_i(a) + f_i(b) + f_i(c) = 0$.
\end{enumerate}
The third condition ensures that in this case $a + b + c = 0$. Moreover, if there is a 3SUM solution in our input set, then indeed for some $(i, j) \in [N] \times [t]$ the above three conditions hold, and we will detect a solution.

We now describe how to generate the Local Alignment instances for each pair $(i, j)$. As the alphabet take $\Sigma = [m] \times [-q, q] \times [d] \cup \set{\$}$ which indeed satisfies $|\Sigma| = \Order(m q d) = \widetilde\Order(n^{1-\delta})$. Focus next on the of the strings $X$ and $Y$. For $a \in A$ and $b \in B$ let
\def\arraystretch{1.2}
\begin{equation*}
    \begin{array}{@{}r@{\;}c@{\;}c@{\;}c@{\;}c@{\;}l}
        X_a & = & (a \bmod m, f_i(a)[1], 1) & \ldots & (a \bmod m, f_i(a)[d], d) & \in \Sigma^d, \\
        Y_b & = & (b \bmod m, f_i(b)[1], 1) & \ldots & (b \bmod m, f_i(b)[d], d) & \in \Sigma^d. \\
    \end{array}
\end{equation*}
Then, writing $A = \set{a_1, \dots, a_n}$ and $B = \set{b_1, \dots, b_n}$, choose
\begin{equation*}
    \begin{array}{@{}c@{\;}c@{\;}c@{\;}c@{\;}c@{\;}c@{\;}c@{\;}c@{\;}c@{\;}l}
        X & = & X_1 & \$ & X_2 & \$ & \ldots & \$ & X_n & \in \Sigma^{nd + n-1}, \\
        Y & = & Y_1 & \$ & Y_2 & \$ & \ldots & \$ & Y_n & \in \Sigma^{nd + n-1}. \\
    \end{array}
\end{equation*}
The scoring function $s(\cdot, \cdot)$ is defined as follows. For two letters $(h_1, y_1, \ell_1), (h_2, y_2, \ell_2) \in \Sigma$, the scoring function looks up $c \in C$ such that $c$ is the $j$-th element in the $(-(h_1 + h_2) \bmod m)$-th bucket, and returns a score of 1 if and only if $\ell_1 = \ell_2$ and \mbox{$f_i(c)[\ell_1] = -(y_1 + y_2)$}. Otherwise, the score is $-\infty$ (or some sufficiently large negative integer). Moreover, we define all scores involving the padding symbol $\$$ as $-\infty$.

\subparagraph{Correctness.}
We claim that the optimal value of any such instance is exactly $d$ if and only if the previous three conditions~(i),~(ii) and~(iii) are satisfied. By the construction it is clear that the optimal value can never exceed $d$, as any substring of length more than $d$ contains at least one padding symbol~$\$$. Thus, the value is exactly $d$ if and only if there exist $a \in A$ and $b \in B$ such that~$X_a$ and~$Y_b$ match perfectly under our scoring function. Let $c \in C$ denote the $j$-th element in the $(-(a + b) \bmod m)$-th bucket. Then the match is perfect if and only if $f_i(c)[\ell] = -f_i(a)[\ell] - f_i(b)[\ell]$ for all $\ell \in [d]$. This completes the correctness proof.

\subparagraph{Running Time.}
It remains to analyze the running time. Brute-forcing the heavy buckets takes time $\Order(n^{2-\gamma})$. Afterwards, recall that we constructed $N \cdot t = n^{\delta + \gamma + \order(1)}$ instances of size $\widetilde\Order(n)$ (and alphabet~$\widetilde\Order(n^{1 - \delta})$). For each such instance we can precompute the scoring tables in time~$\widetilde\Order(n^{2-2\delta})$. Assuming that there is an algorithm solving a single such instance in time $\Order(n^{2-\delta-\epsilon})$, for some~\makebox{$\epsilon > 0$}, the total running time becomes
\begin{equation*}
    \widetilde\Order(n^{2-\gamma} + n^{\delta + \gamma + \order(1)} \cdot (n^{2 - 2\delta} + n^{2 - \delta - \epsilon})) = n^{\max\set{2 - \gamma, 2 - \epsilon + \gamma} + \order(1)},
\end{equation*}
using that $\epsilon \leq \delta$ (due to the trivial lower bound that any algorithm has to read the input). By choosing $\gamma := \frac{\epsilon}{2}$, this running time becomes truly subquadratic.
\end{proof}
\section{Deterministic Lower Bound for Convolution Witness} \label{sec:conv-witness}
In this section we prove that the known randomized hardness reduction for the \emph{Convolution Witness} problem~\cite{GoldsteinKLP16} can be derandomized.

\begin{problem}[Convolution Witness]
Preprocess two vectors $X, Y \in \set{0, 1}^n$ such that upon query of an index $k$ the list of all \emph{witnesses} $\set{(i, j) \in [n] : i + j = k,\, X[i] = Y[i] = 1}$ is returned. 
\end{problem}

\begin{theorem}[Deterministic Convolution Witness Hardness] \label{thm:conv-witness}
Let $0 < \alpha < 1$. Unless the deterministic 3SUM hypothesis fails, there is no algorithm for the Convolution Witness problem with preprocessing time $\Order(n^{2-\alpha})$ and query time $\Order(n^{\alpha/2-\epsilon})$ (per witness), for any $\epsilon > 0$.
\end{theorem}
\begin{proof}
We assume that there is a data structure for the Convolution Witness problem with preprocessing time $\Order(n^{2-\alpha})$ and query time $\Order(n^{\alpha/2-\epsilon})$ per witness. We design a subquadratic-time algorithm for a given 3SUM instance $A \subseteq [n^c]$.

Let $\delta > 0$ be a parameter to be determined later, and let $\mu_1 = 1, \mu_2 = \alpha/2 - \delta$. We apply the deterministic additive hashing lemma (\cref{lem:hashing-few-pseudo-solutions}) two times to obtain moduli $m_1 = \Theta(n^{\mu_1}) = \Theta(n)$ and $m_2 = \Theta(n^{\mu_2}) = \Theta(n^{\alpha/2-\delta})$ with
\begin{equation*}
    S := \#\set{(a, b, c) \in A^3 : a + b \equiv c \mod{m_1 m_2}} \leq \widetilde\Order(n^{3 - \mu_1 - \mu_2}) = \widetilde\Order(n^{2 - \alpha/2 + \delta}).
\end{equation*}
Next, for each pair $x, y \in [m_2]$ we construct the length-$2m_1$ vectors $X_x$ and $Y_y$ as follows:
\begin{align*}
    X_x[i] &=
    \begin{cases}
        1 &\text{if $\exists a \in A$ with $a \equiv i \mod{m_1}$ and $a \equiv x \mod{m_2}$}, \\
        0 &\text{otherwise,}
    \end{cases} \\
    Y_y[j] &=
    \begin{cases}
        1 &\text{if $\exists b \in A$ with $b \equiv j \mod{m_1}$ and $b \equiv y \mod{m_2}$}, \\
        0 &\text{otherwise.}
    \end{cases}
\end{align*}
For each pair $x, y \in [m_2]$ we preprocess the vectors $(X_x, Y_y)$ with the efficient data structure. Then, for each $c \in A$ with $c \equiv (x + y) \mod{m_2}$ we query all witnesses of $k := c \bmod{m_1}$ and $k + m_1$. For each witness $(i, j)$ that is reported, we enumerate all pairs $(a, b) \in A^2$ with $a \equiv i \mod{m_1}$ and $b \equiv j \mod{m_2}$, and test whether $a + b \in A$. In the positive case, we have detected a 3SUM solution and we stop.

\subparagraph{Correctness.}
It is easy to see that this algorithm is correct: Throughout the execution, we enumerating exactly all triples $(a, b, c)$ for which $a + b \equiv c \mod{m_1}$ and $a + b \equiv \mod{m_2}$. In other words, we enumerate exactly all pseudo-solutions modulo $m_1 \cdot m_2$. Since any proper 3SUM solution is in particular a pseudo-solution, we are guaranteed to find it. Moreover, we only report valid solutions.

\subparagraph{Running Time.}
Let us finally bound the running time. First, invoking \cref{lem:hashing-few-pseudo-solutions} runs in time~\makebox{$\Order(n^{\max(1, \mu_1 + \mu_2) + \delta}) = \Order(n^{1 + \alpha/2 + \delta})$}. The remaining algorithm runs for $m_2^2$ iterations, where in each iteration we preprocess two vectors of length $m_1$. By the initial assumption, this takes time~\makebox{$\Order(n^{2\mu_2} \cdot n^{\mu_1 (2 - \alpha)}) = \Order(n^{2-2\delta})$}. Finally, let us consider the query cost. Recall that the total number of witnesses queried is at most $S$---the number of pseudo-solutions modulo $m_1 \cdot m_2$. Each query runs in time $\Order(n^{\alpha/2-\epsilon})$, hence the total query time is $\Order(S \cdot n^{\alpha/2-\epsilon}) = \widetilde\Order(n^{2-\epsilon + \delta})$. The total time is
\begin{equation*}
    \widetilde\Order(n^{1+\alpha/2+\delta} + n^{2-2\delta} + n^{2-\epsilon+\delta}),
\end{equation*}
so by picking $\delta := \min(\frac{\epsilon}{2}, \frac{\alpha}{4}) > 0$ the running time becomes truly subquadratic. This contradicts the deterministic 3SUM hypothesis.
\end{proof}

\end{document}